\documentclass[11pt,a4paper]{article}

\title{The Sherali-Adams and Weisfeiler-Leman hierarchies in (Promise Valued) Constraint Satisfaction Problems\thanks{Extended abstracts of parts of this work appeared in \emph{Proceedings of the 46th International Symposium on Mathematical Foundations of Computer Science} (MFCS 2021) \cite{ButtiD21fractional} (Butti and Dalmau) and in \emph{Proceedings of the 28th International Conference on Principles and Practice of Constraint Programming} (CP 2022) \cite{BartoButti2022} (Barto and Butti). This work is also based on S. Butti's PhD thesis \cite{Butti2022phd}. 

Libor Barto was funded by the European Union (ERC, CoCoSym, 771005) and (ERC, POCOCOP, 101071674). Views and opinions expressed are however those of the authors only and do not necessarily reflect those of the European Union or the European Research Council Executive Agency. Neither the European Union nor the granting authority can be held responsible for them.
Silvia Butti was supported by UKRI EP/X024431/ and by a fellowship from ``la Caixa'' Foundation (ID 100010434). The fellowship code is LCF/BQ/DI18/11660056.
This project has received funding from the European Union’s Horizon 2020 research and innovation programme under the Marie Skłodowska-Curie grant agreement No. 713673.
Victor Dalmau was supported by the MICIN under grants PID2019-109137GB-C22 and PID2022-138506NB-C22, and the Maria de Maeztu program (CEX2021-001195-M).
For the purpose of Open Access, the authors have applied a CC BY public copyright licence to any Author Accepted Manuscript version arising from this submission. All data is provided in full in the results section of this paper.}}

\author{Libor Barto\\
Department of Algebra\\
Charles University, Czechia\\
\texttt{libor.barto@gmail.com}
\and
Silvia Butti\\
Department of Computer Science\\
University of Oxford, UK\\
\texttt{silvia.butti@cs.ox.ac.uk}
\and
V{\'{\i}}ctor Dalmau\\
Department of Information and Communication Technologies\\
Universitat Pompeu Fabra, Spain\\
\texttt{victor.dalmau@upf.edu}
}

\date{}

\usepackage[margin=1in]{geometry}
\usepackage[utf8]{inputenc}
\usepackage[T1]{fontenc} 
\usepackage{amsmath,amssymb,amsthm,mathtools}
\usepackage{subcaption} 

\usepackage{xcolor}
\definecolor{pigment}{rgb}{0.1, 0.1, 0.6}
\usepackage[backref=page]{hyperref} 
\hypersetup{colorlinks=true,citecolor=pigment,linkcolor=pigment,urlcolor=black} 
\usepackage[framemethod=TikZ]{mdframed} 
\usepackage{enumerate} 
\usepackage{thmtools} 
\usepackage{thm-restate} %
\usepackage{mathrsfs} 
\usepackage{tikz,tikz-cd}
\usetikzlibrary{decorations.pathmorphing} 
\usepackage{adjustbox} 

\usepackage{setspace}

\usepackage{etoolbox}
\makeatletter
\patchcmd{\BR@backref}{\newblock}{\newblock(}{}{}
\patchcmd{\BR@backref}{\par}{)\par}{}{}
\makeatother

\newcommand{\con}[1]{\mathcal{C}_{#1}}
\newcommand{\sig}{\sigma} 
\newcommand{\matleft}{P} 
\newcommand{\matright}{Q}
\DeclareMathAlphabet\mathbfcal{OMS}{cmsy}{b}{n}
\newcommand{\partp}{\mathbfcal{P}}
\newcommand{\partq}{\mathbfcal{Q}}
\newcommand{\ipartp}{\mathcal{P}}
\newcommand{\ipartq}{\mathcal{Q}}
\newcommand{\LAB}{\mathcal{L}}
\newcommand{\polytope}{\mathbb{P}}
\newcommand{\spo}{\mathrm{Spoiler}}
\newcommand{\dup}{\mathrm{Duplicator}}

\newcommand{\ba}{\textbf{a}}
\newcommand{\bb}{\textbf{b}}
\newcommand{\bc}{\textbf{c}}
\newcommand{\bd}{\textbf{d}}
\newcommand{\bi}{\textbf{i}}
\newcommand{\bp}{\textbf{p}}

\newcommand{\bu}{\textbf{u}}
\newcommand{\bv}{\textbf{v}}
\newcommand{\bx}{\textbf{x}}
\newcommand{\by}{\textbf{y}}
\newcommand{\bz}{\textbf{z}}

\newcommand{\bfa}{\mathbf{A}}
\newcommand{\bfb}{\mathbf{B}}
\newcommand{\bfc}{\mathbf{C}}
\newcommand{\bfd}{\mathbf{D}}

\newcommand{\bfx}{\mathbf{X}}
\newcommand{\bfy}{\mathbf{Y}}
\newcommand{\yone}{\mathbf{Y}_{1}}
\newcommand{\ytwo}{\mathbf{Y}_{2}}

\newcommand{\graphs}[3]{\begin{tikzpicture}
        \node[circle, draw=black, fill=black, inner sep=0pt, minimum size=#3cm] at (-3*#2,0*#2) (n1) {};
        \node[circle, draw=black, fill=white, inner sep=0pt, minimum size=#3cm] at (-1*#2,0*#2) (n2) {};
        \node[circle, draw=black, fill=black, inner sep=0pt, minimum size=#3cm] at (1*#2,0*#2) (n3) {};
        \node[circle, draw=black, fill=black, inner sep=0pt, minimum size=#3cm] at (-3*#2,2*#2) (n4) {};
        \node[circle, draw=black, fill=white, inner sep=0pt, minimum size=#3cm] at (-1*#2,2*#2) (n5) {};
        \node[circle, draw=black, fill=black, inner sep=0pt, minimum size=#3cm] at (1*#2,2*#2) (n6) {};
    \draw (n1)--(n2);
    \draw (n2)--(n3);
    \draw (n4)--(n5);
    \draw (n5)--(n6);
    \draw (n1)--(n4);
    \draw (n2)--(n5);
    \draw (n3)--(n6);
    \begin{scope}[xshift=#1cm] 
        \node[circle, draw=black, fill=black, inner sep=0pt, minimum size=#3cm] at (-2*#2,0*#2) (n1) {};
        \node[circle, draw=black, fill=white, inner sep=0pt, minimum size=#3cm] at (-0.5*#2,1*#2) (n2) {};
        \node[circle, draw=black, fill=black, inner sep=0pt, minimum size=#3cm] at (2.5*#2,0*#2) (n3) {};
        \node[circle, draw=black, fill=black, inner sep=0pt, minimum size=#3cm] at (-2*#2,2*#2) (n4) {};
        \node[circle, draw=black, fill=white, inner sep=0pt, minimum size=#3cm] at (1*#2,1*#2) (n5) {};
        \node[circle, draw=black, fill=black, inner sep=0pt, minimum size=#3cm] at (2.5*#2,2*#2) (n6) {};
    \draw (n1)--(n2);
    \draw (n1)--(n4);
    \draw (n5)--(n6);
    \draw (n2)--(n5);
    \draw (n3)--(n5);
    \draw (n3)--(n6);
    \draw (n4)--(n2);
    \end{scope}
\end{tikzpicture}}

\newcommand{\trees}[3]{\begin{tikzpicture}
        \node[circle, draw=gray, fill=gray, inner sep=0pt, minimum size=#3cm] at (0*#2,3.5*#2) (root) {};
        \node[circle, draw=gray, fill=gray, inner sep=0pt, minimum size=#3cm] at (-1.25*#2,2.25*#2) (l) {};
        \node[circle, draw=gray, fill=gray, inner sep=0pt, minimum size=#3cm] at (0*#2,2*#2) (c) {};
        \node[circle, draw=gray, fill=gray, inner sep=0pt, minimum size=#3cm] at (1.25*#2,2.25*#2) (r) {};
        \node[circle, draw=gray, fill=gray, inner sep=0pt, minimum size=#3cm] at (-2*#2,1*#2) (ll) {};
        \node[circle, draw=gray, fill=gray, inner sep=0pt, minimum size=#3cm] at (-1.25*#2,0.75*#2) (lr) {};
        \node[circle, draw=gray, fill=gray, inner sep=0pt, minimum size=#3cm] at (-0.5*#2,1*#2) (cl) {};
        \node[circle, draw=gray, fill=gray, inner sep=0pt, minimum size=#3cm] at (0*#2,0.75*#2) (cc) {};
        \node[circle, draw=gray, fill=gray, inner sep=0pt, minimum size=#3cm] at (0.5*#2,1*#2) (cr) {};
        \node[circle, draw=gray, fill=gray, inner sep=0pt, minimum size=#3cm] at (1.25*#2,0.75*#2) (rl) {};
        \node[circle, draw=gray, fill=gray, inner sep=0pt, minimum size=#3cm] at (2*#2,1*#2) (rr) {};

        \node[circle, draw=white, fill=white, inner sep=0pt, minimum size=#3cm] at (-2.25*#2,0*#2) (lll) {};
        \node[circle, draw=white, fill=white, inner sep=0pt, minimum size=#3cm] at (-2*#2,-0.1*#2) (llc) {};
        \node[circle, draw=white, fill=white, inner sep=0pt, minimum size=#3cm] at (-1.75*#2,0*#2) (llr) {};
        \node[circle, draw=white, fill=white, inner sep=0pt, minimum size=#3cm] at (-1.5*#2,0*#2) (lrl) {};
        \node[circle, draw=white, fill=white, inner sep=0pt, minimum size=#3cm] at (-1.25*#2,-0.1*#2) (lrr) {};
        \node[circle, draw=white, fill=white, inner sep=0pt, minimum size=#3cm] at (-0.25*#2,0*#2) (ccl) {};
        \node[circle, draw=white, fill=white, inner sep=0pt, minimum size=#3cm] at (0*#2,-0.1*#2) (ccc) {};
        \node[circle, draw=white, fill=white, inner sep=0pt, minimum size=#3cm] at (0.25*#2,0*#2) (ccr) {};
        \node[circle, draw=white, fill=white, inner sep=0pt, minimum size=#3cm] at (-0.75*#2,0*#2) (cll) {};
        \node[circle, draw=white, fill=white, inner sep=0pt, minimum size=#3cm] at (-0.5*#2,-0.1*#2) (clr) {};
        \node[circle, draw=white, fill=white, inner sep=0pt, minimum size=#3cm] at
        (0.5*#2,-0.1*#2) (crl) {};
        \node[circle, draw=white, fill=white, inner sep=0pt, minimum size=#3cm] at
        (0.75*#2,0*#2) (crr) {};
        \node[circle, draw=white, fill=white, inner sep=0pt, minimum size=#3cm] at
        (1.25*#2,-0.1*#2) (rll) {};
        \node[circle, draw=white, fill=white, inner sep=0pt, minimum size=#3cm] at
        (1.5*#2,0*#2) (rlr) {};
        \node[circle, draw=white, fill=white, inner sep=0pt, minimum size=#3cm] at
        (1.75*#2,0*#2) (rrl) {};
        \node[circle, draw=white, fill=white, inner sep=0pt, minimum size=#3cm] at
        (2*#2,-0.1*#2) (rrc) {};
        \node[circle, draw=white, fill=white, inner sep=0pt, minimum size=#3cm] at
        (2.25*#2,0*#2) (rrr) {};
    \draw (root)--(l);
    \draw (root)--(c);
    \draw (root)--(r);
    \draw (l)--(ll);
    \draw (l)--(lr);
    \draw (r)--(rl);
    \draw (r)--(rr);
    \draw (c)--(cr);
    \draw (c)--(cl);
    \draw (c)--(cc);

    \draw (ll)--(lll);
    \draw (ll)--(llc);
    \draw (ll)--(llr);
    \draw (lr)--(lrl);
    \draw (lr)--(lrr);
    \draw (cl)--(cll);
    \draw (cl)--(clr);
    \draw (cc)--(ccl);
    \draw (cc)--(ccc);
    \draw (cc)--(ccr);
    \draw (cr)--(crl);
    \draw (cr)--(crr);
    \draw (rl)--(rll);
    \draw (rl)--(rlr);
    \draw (rr)--(rrl);
    \draw (rr)--(rrc);
    \draw (rr)--(rrr);
    \begin{scope}[xshift=#1cm] 
      \node[circle, draw=gray, fill=gray, inner sep=0pt, minimum size=#3cm] at (0*#2,3.5*#2) (root) {};
        \node[circle, draw=gray, fill=gray, inner sep=0pt, minimum size=#3cm] at (-1*#2,2.25*#2) (l) {};
        \node[circle, draw=gray, fill=gray, inner sep=0pt, minimum size=#3cm] at (1*#2,2.25*#2) (r) {};
        \node[circle, draw=gray, fill=gray, inner sep=0pt, minimum size=#3cm] at (-1.75*#2,1*#2) (ll) {};
        \node[circle, draw=gray, fill=gray, inner sep=0pt, minimum size=#3cm] at (-1*#2,0.75*#2) (lr) {};
        \node[circle, draw=gray, fill=gray, inner sep=0pt, minimum size=#3cm] at (0.25*#2,1*#2) (rl) {};
        \node[circle, draw=gray, fill=gray, inner sep=0pt, minimum size=#3cm] at (1*#2,0.75*#2) (rc) {};
        \node[circle, draw=gray, fill=gray, inner sep=0pt, minimum size=#3cm] at (1.75*#2,1*#2) (rr) {};

        \node[circle, draw=white, fill=white, inner sep=0pt, minimum size=#3cm] at (-2*#2,0*#2) (lll) {};
        \node[circle, draw=white, fill=white, inner sep=0pt, minimum size=#3cm] at (-1.75*#2,-0.1*#2) (llr) {};
        \node[circle, draw=white, fill=white, inner sep=0pt, minimum size=#3cm] at (-1.25*#2,0*#2) (lrl) {};
        \node[circle, draw=white, fill=white, inner sep=0pt, minimum size=#3cm] at (-1*#2,-0.1*#2) (lrc) {};
        \node[circle, draw=white, fill=white, inner sep=0pt, minimum size=#3cm] at (-0.75*#2,0*#2) (lrr) {};

        \node[circle, draw=white, fill=white, inner sep=0pt, minimum size=#3cm] at (0*#2,-0.1*#2) (rll) {};
        \node[circle, draw=white, fill=white, inner sep=0pt, minimum size=#3cm] at (0.25*#2,0*#2) (rlr) {};
        \node[circle, draw=white, fill=white, inner sep=0pt, minimum size=#3cm] at (0.75*#2,0*#2) (rcl) {};
        \node[circle, draw=white, fill=white, inner sep=0pt, minimum size=#3cm] at (1*#2,-0.1*#2) (rcc) {};
        \node[circle, draw=white, fill=white, inner sep=0pt, minimum size=#3cm] at (1.25*#2,0*#2) (rcr) {};
        \node[circle, draw=white, fill=white, inner sep=0pt, minimum size=#3cm] at (1.75*#2,-0.1*#2) (rrl) {};
        \node[circle, draw=white, fill=white, inner sep=0pt, minimum size=#3cm] at (2*#2,0*#2) (rrr) {};

    \draw (root)--(l);
    \draw (root)--(r);
    \draw (l)--(ll);
    \draw (l)--(lr);
    \draw (r)--(rl);
    \draw (r)--(rc);
    \draw (r)--(rr);

    \draw (ll)--(lll);
    \draw (ll)--(llr);
    \draw (lr)--(lrl);
    \draw (lr)--(lrc);
    \draw (lr)--(lrr);
    \draw (rl)--(rll);
    \draw (rl)--(rlr);
    \draw (rc)--(rcl);
    \draw (rc)--(rcc);
    \draw (rc)--(rcr);
    \draw (rr)--(rrl);
    \draw (rr)--(rrr);

    \end{scope}
\end{tikzpicture}}

\newcommand{\pr}{\pi}
\newcommand{\na}{N_{\bfa}}
\newcommand{\nb}{N_{\bfb}}
\newcommand{\ngr}{N_{G}}
\newcommand{\nh}{N_{H}}
\newcommand{\eqwl}[1][1]{\equiv_{#1}}
\newcommand{\stark}{*_{k}}
\newcommand{\xstarkbf}{\mathbf{X}^\stark}
\newcommand{\astarkbf}{\mathbf{A}^\stark}
\newcommand{\bstarkbf}{\mathbf{B}^\stark}

\newcommand{\astark}{A^\stark}
\newcommand{\bstark}{B^\stark}
\newcommand{\sak}{\textnormal{SA}^k(\bfx,\bfa)}
\newcommand{\sastar}{\textnormal{SA}^1(\xstarkbf,\astarkbf)} 
\newcommand{\pii}{\pi_{\textbf{i}}}
\newcommand{\tjii}{T_{j,\textbf{i}}}

\newcommand{\csp}{\textnormal{CSP}}
\newcommand{\cspa}{\csp(\bfa)}
\newcommand{\pcsp}{\textnormal{PCSP}}
\newcommand{\pvcsp}{\textnormal{PVCSP}}
\newcommand{\sa}{\textnormal{SA}}
\newcommand{\blp}{\textnormal{BLP}}
\newcommand{\val}{\operatorname{Val}}
\newcommand{\opt}{\operatorname{Opt}}
\newcommand{\ar}{\operatorname{ar}}
\newcommand{\yes}{\mathrm{Yes}}
\newcommand{\no}{\mathrm{No}}

\newcommand{\ab}{(\bfa,\bfb)}
\newcommand{\xa}{(\bfx,\bfa)}
\newcommand{\xb}{(\bfx,\bfb)}

\newcommand{\saxa}{\sa^{1}\xa}
\newcommand{\qq}{\mathbb{Q}_{\geq0}}
\newcommand{\qinfty}{\mathbb{Q}_{\infty}}
\newcommand{\graph}[1]{\mathrm{G}_{#1}} 

\newcommand{\thresh}{\tau}
\newcommand{\RA}{R^\bfa}
\newcommand{\RB}{R^\bfb}

\newcommand{\eqone}{\equiv_{1}}
\newcommand{\ith}{^{\textnormal{th}}}
\newcommand{\mult}[1]{\{\!\{#1\}\!\}}
\def\multiset#1#2{\ensuremath{\left(\kern-.3em\left(\genfrac{}{}{0pt}{}{#1}{#2}\right)\kern-.3em\right)}}
\newcommand{\counting}{\mathscr{C}}

\newtheorem{theorem}{Theorem}
\newtheorem{corollary}[theorem]{Corollary}
\newtheorem{proposition}[theorem]{Proposition}
\newtheorem{lemma}[theorem]{Lemma}

\theoremstyle{definition}

\theoremstyle{remark}

\newtheorem{example}{Example}

\begin{document}

\maketitle

\begin{abstract}

In this paper we study the interactions between so-called fractional relaxations of the integer programs (IPs) which encode homomorphism and isomorphism of relational structures.
We give a combinatorial characterization
of a certain natural linear programming (LP) relaxation of homomorphism in terms of fractional isomorphism. As a result, we show that the families of constraint satisfaction problems (CSPs) that are solvable by such linear program are precisely those that are closed under an equivalence relation which we call \textit{Weisfeiler-Leman invariance}. We also generalize this result to the much broader framework of Promise Valued Constraint Satisfaction Problems, which brings together two well-studied extensions of the CSP framework.
Finally, we consider the hierarchies of increasingly tighter relaxations of the homomorphism and isomorphism IPs obtained by applying the Sherali-Adams and Weisfeiler-Leman methods respectively. We extend our combinatorial characterization
of the basic LP to higher levels of the Sherali-Adams hierarchy, and we generalize a well-known logical characterization of the Weisfeiler-Leman test from graphs to relational structures.

\end{abstract}

{
  \hypersetup{linkcolor=black}
}

\newpage
\section{Introduction}

\paragraph{Isomorphism, homomorphism, and fractional relaxations.}
The graph isomorphism and homomorphism problems, that is, the problems of deciding the existence of an isomorphism (respectively homomorphism) from a given input graph to another, have been the subject of extensive research. 
Despite an important research effort, it is still an open problem to determine whether the isomorphism problem for graphs can be solved in polynomial time. The recent quasipolynomial algorithm for the problem presented by Babai \cite{Babai2015} is widely regarded as a major breakthrough in theoretical computer science.

In turn, the complexity of the graph homomorphism problem has also been intensively studied in combinatorics (see \cite{HellNesetril:book}) as well as in a more general setting, known as the Constraint Satisfaction Problem (CSP), where the inputs are not required to be graphs but can be arbitrary relational structures. The CSP 
is general enough to encompass problems from areas as diverse as artificial intelligence, optimization, computer algebra, computational biology, computational linguistics, among many other. In contrast with the graph isomorphism problem it was quickly established that the homomorphism problem is NP-complete even for graphs, as it can encode computationally hard problems such as graph coloring.
Consequently, an important research effort has been put into identifying tractable fragments of the problem, in particular by fixing the target structure. This culminated in the recent major result of Bulatov \cite{Bulatov2017} and Zhuk \cite{Zhuk2020}, which confirmed a conjecture by Feder and Vardi \cite{feder1998computational} that, for a fixed target structure, the homomorphism problem is either solvable in polynomial time or it is NP-complete -- in particular, there are no NP-intermediate CSPs.

Linear programming relaxations, among other relaxations such as SDP-based, have been largely used in the study of both the isomorphism and the homomorphism problem. In fact, the isomorphism problem for graphs $G$, $H$ can be reformulated as an integer program which asks whether there exists a permutation matrix $\matleft$ such that $\matleft N_G=N_H \matleft$, where $N_G$ and $N_H$ are the adjacency matrices of $G$, $H$ respectively. If we relax this condition to only require that $\matleft$ is doubly stochastic, we obtain what is known as fractional isomorphism.

 Now, this linear program can be equivalently expressed as the existence of a pair of doubly stochastic matrices $\matleft$ and $\matright$ such that $\matleft M_G=M_H \matright$ and $M_H \matright^T=\matleft^T M_G$, where $M_G$ now denotes the \textit{incidence} matrix of $G$. If we relax this condition further to only require that $\matleft$ and $\matright$ are \textit{left} stochastic and, additionally, we drop the second equation, then we obtain a relaxation of graph homomorphism known as the \textit{Basic Linear Programming} relaxation.\footnote{We remark that in \cite{ButtiD21fractional}, this algebraic condition - and its equivalent characterizations - was alternatively phrased as the existence of a fractional homomorphism, to stress that it is a fractional relaxation of homomorphism in the same way that fractional isomorphism is a relaxation of isomorphism. Nonetheless, in this paper we avoid this terminology as it clashes with the notion of fractional homomorphism defined as a unary fractional polymorphism, which is standard terminology in the context of Valued CSPs (see Section \ref{sec:fr-ops}).} With some minor variations depending on whether the objective function is present or not and how repeated elements in a tuple are treated, this LP formulation has been extensively used \cite{KumarMTV11,Kun2012,DalmauK13,DalmauKM18,GhoshT18,BrakensiekGWZ20}.

 The goal of this paper is to study the interaction between the fractional relaxations of homomorphism and isomorphism introduced above, as well as their higher-dimensional enhancements, in the context of the Constraint Satisfaction Problem and its extension to the Promise Valued setting.

\paragraph{The WL and SA hierarchies.}
Fractional isomorphism has a combinatorial counterpart in the 1-dimensional Weisfeiler-Leman (1-WL) algorithm \cite{leman1968reduction}, also known as colour refinement. In particular, it was shown in \cite{Tinhofer1986, Tinhofer1991} and \cite{Ramana1994} that two graphs $G$ and $H$ are fractionally isomorphic if and only if 1-WL does not distinguish between them. 1-WL produces a sequence of colourings $c_0,c_1,\ldots$ of the nodes of a graph by means of an iterative refinement procedure which assigns a pair of nodes to the same colour class of $c_i$ if they belong to the same class of $c_{i-1}$, and additionally they have the same number of neighbours of each colour in $c_{i-1}$. The algorithm keeps iterating until a fixed point is reached. The Weisfeiler-Leman algorithm is a very powerful heuristic to test for graph isomorphism: if two graphs are distinguished by 1-WL (that is, they give rise to distinct fixed-point colourings up to renaming of the vertices), then this is a witness that the graphs are not isomorphic. In fact, it was shown that 1-WL decides the isomorphism problem on almost all graphs \cite{BabaiKucera79,Babai1980}. 

However, it is also easy to see that 1-WL fails on some very simple instances, such as regular graphs. 
To address these limitations, the original Weisfeiler-Leman algorithm has been extended so that at every iteration on a graph $G$ it produces a colouring of the set of $k$-tuples (for $k > 1$) of nodes of $G$. This higher-dimensional version of the Weisfeiler-Leman algorithm ($k$-WL) turned out to be very robust: it has very neat equivalent formulations in logical and combinatorial terms \cite{Cai92,Dvorak2010,Dell2018}, and has found important applications not just in the context of the graph isomorphism problem \cite{Kiefer2020wlpower,GroheSchweitzer2020survey,Grohe2021isomorphism} but also within the recently developed theory of Graph Neural Networks \cite{XuHowPowerful,morris2019weisfeiler,Morris2021WLML}.

Similarly to the Weisfeiler-Leman hierarchy for fractional isomorphism, any LP relaxation of an integer $0/1$-program 
can be strengthened by sequentially applying so-called lift-and-project methods from mathematical programming in order to obtain a hierarchy of increasingly tighter relaxations which can still be solved efficiently. The main idea behind these is to add auxiliary variables and valid inequalities to an initial relaxation of a $0/1$ integer program.
These methods, which include Lov\'asz-Schrijver \cite{Lovasz1991} and Sherali-Adams \cite{Sherali1990}, have been used to study classical problems in combinatorial optimization such as Max-Cut, Vertex Cover, and Maximal Matching among many others.

Let us now turn our attention to the homomorphism problem. Here, LP relaxations have been intensively used in the more general setting of the Constraint Satisfaction Problem and, most usually, in the approximation of its optimization versions such as MaxCSP, consisting of finding a map from the input structure to the target structure that maximizes the number of constraints satisfied. For instance, a simple linear programming relaxation yields a $2$-approximation algorithm for the Vertex Cover problem, and no better polynomial time approximation algorithm is known.

We consider relaxations arising from the application of the Sherali-Adams (SA) method, which coincides with the hierarchy commonly used for MaxCSP \cite{ChanLRS13,GeorgiouMT09,ThapperZ17,YoshidaZ14} where the objective function has been turned into a constraint. 

The power of the SA hierarchy is well-studied in the context of the CSP and its approximation (Promise CSP) \cite{Barto2021algebraicPCSP,BrakensiekGWZ20,BrakensiekG2020symmetric,atserias2022,Ciardo2022Tensors,DalmauOprsal2023} and optimisation (Valued CSP) \cite{ThapperZ12,KolmogorovTZ15,ThapperZ16,ThapperZ17} variants. Moreover, a surprising connection between the Sherali-Adams and Weisfeiler-Leman hierarchies has been found in the context of the graph isomorphism problem \cite{Atserias2013,Malkin2014}, showing that the two hierarchies interleave in power. We study this connection in the context of the CSP and its joint extension to the promise and valued setting, which we introduce next.

\paragraph{Promise Valued CSPs.}

The recently introduced framework of the Promise Constraint Satisfaction Problem (PCSP) was originally motivated by open questions in the theory of approximation \cite{AGH17,BrakensiekG21booleanDichotomy}. Intuitively, in this setting each
constraint comes in two forms, one strong and one weak, and the task is to distinguish whether the input is satisfiable in the strong sense, or it is not satisfiable even in the weak sense. Here, the “promise” is that the input will fall precisely into one of these cases. More formally, a PCSP template
is a pair of structures $\ab$ of the same signature, and the PCSP over $\ab$ is the problem of distinguishing structures
homomorphic to $\bfa$ from those that are not homomorphic to $\bfb$. Note that when $\bfa$ and $\bfb$ coincide, the PCSP reduces to the CSP.

A well-known family of PCSP examples is the problem of distinguishing $k$-colorable graphs from those that are not even $l$-colorable for some fixed $l \geq k$ \cite{Garey1976coloring}. While a complete complexity classification for PCSPs seems currently
far away, in recent years there have been substantial developments in the understanding of the power of specific hierarchies of algorithms for PCSPs, and particularly for the above-mentioned approximate graph coloring problem  \cite{atserias2022,Ciardo2022Tensors,CiardoZ23crystal,CiardoZ23Hollow}.

The framework of Valued Constraint Satisfaction Problems (VCSP) is a
generalization of the CSP that has an optimisation flavour \cite{KZ_vcsps}.
In the VCSP, instead of relations we consider valued relations, which can be seen as cost functions -- i.e., mappings that assign to tuples rational or positive infinite costs.
An instance is specified by a set of variables and a weighted sum of valued constraints over these variables, known as the objective
function, together with a rational threshold $\thresh$. The task is to decide whether the minimum value of the objective function is at most $\thresh$.

Notice that the VCSP indeed generalizes the CSP since relations can be modelled by $\{0,\infty\}$-valued relations.\footnote{The term \textit{crisp} \cite{Cohen06algvalued} is sometimes used in contrast to \textit{valued} to refer to $\{0,\infty\}$-valued or classical relations.} On the other hand, MaxCSP is exactly the VCSP over $\{0,1\}$-valued relational structures.
The VCSP area is also well developed; in particular, the power of LP-based algorithms is well-understood \cite{ThapperZ12,KolmogorovTZ15,ThapperZ16,ThapperZ17}, and a dichotomy theorem is available~\cite{Kozik2015,Kolmogorov2017}.

Finally, the Promise Valued CSP (PVCSP) combines both generalizations together. A template is a pair of valued structures of the same signature and the problem is, given an objective function and a rational threshold $\thresh$, to distinguish sums whose minimum computed in $\bfa$ is at most $\thresh$ from those whose minimum in $\bfb$ is strictly greater than $\thresh$.

We believe that the PVCSP is an extremely promising research direction for two reasons. First, it is very broad: it includes, for example, all constant-factor approximation problems for MaxCSP (both the version where the aim is to approximately maximize the number of satisfied constraints and the version where the aim is to approximately minimize the number of unsatisfied constraints). Second, the algebraic approach that was successfully developed to study the complexity of both the promise and the valued CSP can be generalized to the PVCSP as well \cite{Kazda21,AAA24}. 
The only published work on PVCSP that we are aware of (other than \cite{BartoButti2022}) is~\cite{ViolaZ21}, which even considers the more general infinite-domain case.

\subsection{Contributions}

Our contributions can be broadly divided into three parts.
\paragraph{1. A combinatorial decomposition of $\sa^1$.} We start by considering the first level of the Sherali-Adams hierarchy, denoted $\sa^1$, in Section \ref{chap:sawl}. The main result of this section, Theorem \ref{thm:SA1}, is a two-fold characterization of the feasibility of $\sa^1$: algebraic, in terms of the linear program described in the introduction, and combinatorial, relating $\sa^1$ to fractional isomorphism. In particular, we show that for any two similar structures $\bfx$, $\bfa$ the linear program $\saxa$ is feasible if and only if there exists a sequence of structures $\bfy_0,\yone,\dots,\bfy_n$ with $\bfx=\bfy_0$, $\bfa=\bfy_n$, and for every $i<n$, $\bfy_i$ is either homomorphic or fractionally isomorphic to $\bfy_{i+1}$. Moreover, if such a chain of structures exists then there exists a chain of length 3.

Each of these morphisms - homomorphism and fractional isomorphism - from $\bfy_i$ to $\bfy_{i+1}$ can be naturally associated with a rational matrix of dimensions $|Y_i| \times |Y_{i+1}|$. The product of these matrices is in turn a matrix associated to a solution to $\saxa$. This is why we regard this result as a decomposition theorem.

The crux of Theorem \ref{thm:SA1} is the fact that $\sa^1$ is able to certify that $\bfx$ is not homomorphic to $\bfa$ unless $\bfx$ belongs to the backwards closure of $\bfa$ under homomorphism and fractional isomorphism. 
We apply this in Theorem \ref{th:main-csp}, which extends the characterization of CSPs solvable by linear programming with the additional condition of closure under fractional isomorphism, or \textit{Weisfeiler-Leman invariance} as we refer to it in the rest of the paper. This result, which is a direct consequence of the decomposition theorem, answers a question of \cite{Butti2022distr-TOCS} which asks which CSPs can be solved efficiently by a distributed network of anonymous agents which communicate with each other by sending messages through fixed communication channels. 

\paragraph{2. The PVCSP and a decomposition of $\sa^1$ for valued structures.}

In Section \ref{chap:pvcsp} we introduce the PVCSP framework, and in particular we relate PVCSP templates to the notion of fractional homomorphism (Proposition~\ref{prop:FH}), extending the corresponding characterization of PCSP templates. We then lift the decomposition theorem of the previous section from crisp to valued structures (Theorem \ref{thm:Decomp}). One component of this decomposition is a kind of morphism, called here a dual fractional homomorphism, which has appeared previously in the context of VCSPs with left-hand side (i.e., structural) restrictions~\cite{Carbonnel2022otherside}.\footnote{\cite{Carbonnel2022otherside} uses the terminology ``inverse fractional homomorphism'', however we feel that ``dual'' better fits the meaning of this concept.}

Finally, we apply the decomposition theorem for valued structures to obtain a characterization of PVCSPs (and hence of PCSPs and VCSPs as well) solvable by linear programming in terms of Weisfeiler-Leman invariance (Theorem \ref{th:main-pvcsp}).

Surprisingly, two commonly used formulations of the basic linear programming relaxation that are equivalent for CSPs turn out to be substantially different in the broader PVCSP framework. In Example \ref{ex:SAnotBLP} we present a PVCSP template that witnesses this simple but interesting fact.

\paragraph{3. Connecting $\sa^k$ and $k$\textnormal{-WL}.}
In Section \ref{chap:higher}, we consider higher levels of the Sherali-Adams and Weisfeiler-Leman hierarchies.

As for Sherali-Adams, we extend the decomposition theorem from $\sa^1$ to higher levels of the hierarchy (Theorem \ref{thm:SAk}). The correspondence for any higher level $k$ is obtained by replacing 
$\bfx$ and $\bfa$ by suitably enhanced structures $\xstarkbf$, $\astarkbf$ that allow us to reduce the $k^{th}$ level of the hierarchy to the first level (Lemma \ref{le:starkSherali}).

Additionally, replacing two structures $\bfa$ and $\bfb$ in the definition of fractional isomorphism with the corresponding enhanced structures $\astarkbf$, $\bstarkbf$ allows us to define a hierarchy of relaxations of (relational structure) isomorphism, which we denote using $\eqwl[k]$. Our results show that, for every $k\geq 1$, $\sa^k$ is tightly related with the corresponding 
$k\ith$ level relaxation of isomorphism. In particular, feasibility of $\sa^k$ is closed under $\eqwl[k]$-equivalence (see Corollary \ref{cor:sa-k-eqwl-k}).

Finally, a famous result on the Weisfeiler-Leman algorithm is that for the task of distinguishing non-isomorphic graphs, $k$-WL is as powerful as the $(k+1)$-variable fragment of first-order logic with counting quantifier, which in turn is as powerful as counting homomorphisms from graphs of treewidth at most $k$. We extend this result from graph isomorphism to general isomorphism between relational structures (Theorem \ref{th:homTreewidth}). In particular, this shows that $\eqwl[k]$-equivalence correctly extends the notion of higher-dimensional Weisfeiler-Leman indistinguishability from graphs to relational structures.

\paragraph{Paper organisation.}
This paper is organised as follows. In Section  \ref{sec:prelims} we give all basic definitions and set some notation. In Sections \ref{sec:1WL} and \ref{sec:LP} we discuss fractional isomorphism and LP-based relaxations of homomorphism respectively. In Section \ref{chap:sawl} we prove the decomposition theorem for the first level of the Sherali-Adams hierarchy, which contains the core technical construction of this paper, and we apply this to obtain a characterization of Weisfeiler-Leman invariant CSPs. In Section \ref{chap:pvcsp} we extend the decomposition theorem to valued structures and apply this to the broader framework of Promise Valued CSPs. Finally, in Section \ref{chap:higher} we extend the decomposition theorem to higher levels of the Sherali-Adams hierarchy, and we show that the characterization of fractional isomorphism and its higher-dimensional counterpart extends naturally from graphs to relational structures. 
\section{Preliminaries} \label{sec:prelims}

\paragraph{Notation.}
For a natural number $k$, we set $[k] = \{1, \ldots, k\}$. We will denote tuples in boldface. For a set $A$ and a
tuple $\ba \in A^k$, we use either $\ba[i]$ or, where there is no ambiguity, $a_i$ to refer to the $i\ith$ entry of a. We say that $\ba$ has a \textit{repetition} if there exist $i \neq j \in [k]$ such that $a_i = a_j$.

We use double curly brackets $\mult{\dots}$ to denote multisets. For a non-negative integer $k$, $k \cdot \mult{\dots}$ stands for the multiset obtained by multiplying the multiplicity of each element in the original multiset by $k$. Slightly abusing the notation, the set and the multiset of entries of a tuple $\ba$ are denoted by $\{\ba\}$ and $\mult{\ba}$, respectively.

A matrix with non-negative real entries is said to be \textit{left} (resp. \textit{right}) \textit{stochastic} if all its columns (resp. rows) sum to $1$. A \textit{doubly stochastic matrix} is a square matrix that is both left and right stochastic.

\paragraph*{Relational structures.}

A \textit{signature} $\sig$ is a finite collection of relation symbols, each with an associated arity. We shall use $\ar(R)$ to denote the arity of a relation symbol $R$. Given a set $A$ and a positive integer $r$, an $r$-ary \textit{relation} over $A$ is a subset of $A^r$. A \textit{finite relational structure} $\bfa$ over $\sig$, or simply a $\sig$-structure, consists of a finite set $A$ called the \textit{universe} of $\bfa$, and a non-empty relation $\RA$ of arity $\ar(R)$ over $A$ for each $R \in \sig$. We shall use the same boldface and standard capital letter to refer to a structure and its universe, respectively. 
We denote by $\mathcal{C}_\bfa$ the set of formal expressions $\{R(\ba) \mid \ba \in \RA, R \in \sig\}$. Elements of $\mathcal{C}_\bfa$ will be referred to as \textit{constraints}.
Two structures are said to be \emph{similar} if they have the same signature.

 Given two similar structured $\bfa$ and $\bfb$, an \textit{isomorphism} from $\bfa$ to $\bfb$ is a bijective map $f: A \to B$ such that for every $R \in \sig$  and every $\textbf{a} \in A^{\ar(R)}$  it holds that  $\ba \in \RA$ if and only if $f(\textbf{a}) \in \RB$.

The {\em union} $\bfa\cup\bfb$ of two $\sig$-structures $\bfa$ and $\bfb$ is the structure $\bfc$
with $C=A\cup B$ and $R^\bfc=\RA\cup \RB$ for every $R\in\sig$. The {\em disjoint union} of two structures
$\bfa$ and $\bfb$ is the structure $\bfa\cup \bfb'$ where $\bfb'$ is obtained from $\bfb$ by renaming variables wherever necessary so that $A\cap B'=\emptyset$. We say that a structure is {\em connected} if it cannot be expressed as the disjoint union of two structures. We say that $\bfa$ is a {\em substructure} of $\bfb$ if $\bfa\cup\bfb=\bfb$. If, in addition,
$\RA=\RB\cap A^{\ar(R)}$ for every $R\in\sig$ then $\bfa$ is the substructure
of $\bfb$ {\em induced} by $A$.

A \textit{digraph} is a relational structure whose signature consists of a single binary \textit{edge} relation. A \textit{graph} is a digraph where the edge relation is symmetric and non-reflexive. We will often represent graph edges as sets rather than ordered pairs to stress that the edge relation is symmetric. A \textit{cycle} in a graph with edge set $E$ is a sequence of vertices $v_0,v_1,\ldots,v_k$ with $k \geq 3$ such that $\{v_{i-1},v_i\} \in E$ for all $i \in [k]$, $v_0=v_k$, and $\{v_1,\ldots,v_{k}\}$ are distinct. 
A \textit{tree} is a graph that contains no cycles.

The \textit{Factor Graph}\footnote{We note that the notion of factor graph, although similar, differs in several ways from the incidence multigraph (see \cite{LaroseLT07}) as the latter allows for parallel edges.} \cite{Fioretto2018} $\graph{\bfa}$ of a structure $\bfa$ is the undirected labelled bipartite graph with vertex set $A \cup \con\bfa$ and edge set $\{\{a,R(\ba)\} \mid a \in \{\ba\}\}$. Each edge in $\graph{\bfa}$ that is incident to a variable $a$ and a constraint $R(\ba)$ has a label $\ell_{\{a,R(\ba)\}}=(S,R)$ for $S=\{i \in [\ar(R)] \mid a_i=a\}$. We denote the set of labels $\{(S,R) \mid S \subseteq [\ar(R)], R \in \sig\}$ by $\LAB_\sig$.

\paragraph{Constraint Satisfaction Problems.}
Let $\bfx$, $\bfa$ be $\sig$-structures. A \textit{homomorphism} from $\bfx$ to $\bfa$ is a map $h: X \to A$ such that for every $R \in \sig$ and every $\bx \in R^\bfx$ it holds that $h(\bx) \in \RA$, where $h$ is applied to $\bx$ component-wise.
If there exists a homomorphism from $\bfx$ to $\bfa$ we say that $\bfx$ is homomorphic to $\bfa$ and we write $\bfx \to \bfa$. 
We shall use $\hom(\bfx;\bfa)$ to denote the number of homomorphisms from $\bfx$ to $\bfa$.

For a finite relational structure $\bfa$, the (fixed-template, finite-domain) \emph{CSP over $\bfa$}, denoted $\csp(\bfa)$, is the problem of deciding whether a given input structure $\bfx$ similar to $\bfa$ is homomorphic to $\bfa$. 
$\bfx$ is also referred to as the \textit{instance} and $\bfa$ as the \emph{template} in this context.

Let $R \subseteq A^r$ be an $r$-ary relation over $A$. A $k$-ary \textit{polymorphism} of $R$ is an operation $f: A^{k} \to A$ such that the coordinate-wise application of $f$ to any list of $k$ tuples from $R$ gives a tuple in $R$. We say that a function $f$ is a polymorphism of a $\sig$-structure $\bfa$ if $f$ is a polymorphism of $\RA$ for all $R \in \sig$.
A $k$-ary operation $f:A^k \to A$ is said to be \textit{symmetric} if for all $a_1,\ldots,a_k \in A$ and all permutations $\rho$ on $[k]$ we have that $f(a_1,\ldots,a_k) = f(a_{\rho(1)},\ldots,a_{\rho(k)})$.

\section{Combinatorial relaxations of isomorphism} \label{sec:1WL}

We start by describing a number of relaxations for the graph isomorphism problem in Section \ref{sec:friso-graphs}, and in Section \ref{sec:friso-relstr} we will see how to extend these methods to arbitrary relational structures.

\subsection{Fractional isomorphism of graphs} \label{sec:friso-graphs}

Recall that two graphs are isomorphic if there exists an edge-preserving bijection between their vertex sets. As anticipated in the introduction, we can define a weaker notion of similarity between graphs which is known as a \textit{fractional isomorphism}. Two graphs $G$ and $H$ are fractionally isomorphic if there exists a doubly stochastic matrix $\matleft$ such that $\matleft \ngr=\nh \matleft$, where $\ngr$ and $\nh$ are the adjacency matrices of $G$ and $H$ respectively. Fractional isomorphism has a number of equivalent characterizations, which we introduce next.

In a graph $G$, the degree of a vertex $v$ is the number of edges incident to $v$. The zeroth iterated degree of $v$ is equal to its degree. For $j \geq 1$, the $j\ith$ iterated degree of $v$ is the multiset of $(j-1)\ith$ degrees of $v$'s neighbours in $G$. Then, the iterated degree\footnote{Sometimes referred to as the \textit{ultimate} iterated degree \cite{Scheinerman2011fractional}.} of a vertex $v$ is given by the list of its $j\ith$ iterated degrees for all $j \geq 0$, and the iterated degree sequence of a graph $G$ is the multiset of iterated degrees of the vertices of $G$.\footnote{The degree sequence is often defined to be a list. However, when looking at iterated degree it is common~\cite{Ramana1994,Scheinerman2011fractional} and more practical to use multisets instead of lists, while maintaining the terminology \textit{sequence} to highlight that we are dealing with a generalisation of the classical concept of degree sequence.}

One can think of the $j\ith$ iterated degree of a vertex as the ``local view'' of $G$ from $v$. That is, the $j\ith$ iterated degree of $v$ can be seen as the rooted tree $T_v$ of depth $j$ constructed inductively in such a way that there is a map $\gamma$ from the vertices of $T_v$ to the vertices of $G$ which maintains the following properties: the root is mapped by $\gamma$ to $v$, and for every $t$ in $T_v$, the set of children of $t$ is mapped bijectively to the neighbours of $\gamma(t)$. Then, two vertices have the same iterated degree if the corresponding infinite trees constructed via this procedure are isomorphic (see Figure \ref{fig:trees}).

The \textit{colour refinement algorithm}, also known as \textit{1-dimensional Weis\-fei\-ler-Le\-man algorithm} (1-WL)~\cite{leman1968reduction}\footnote{To be precise, we point out that the algorithm proposed in the original paper of Weisfeiler and Leman \cite{leman1968reduction} is equivalent to what is known as the 2-dimensional Weisfeiler-Leman algorithm. This will be defined in its general ($k$-dimensional) version in Section \ref{chap:higher}.}, is the procedure that calculates the iterated degree sequence of a graph (see for instance \cite{grohe2017color}). Since only a linear number of iterations is needed to reach a stable point, 1-WL is often used as a simple heuristic for the graph isomorphism problem: if two graphs are isomorphic then they must have the same iterated degree sequence, (the opposite, however, is not true). We say that 1-WL distinguishes two graphs $G$ and $H$ if their iterated degree sequences differ. Famously, 1-WL distinguishes almost all non-isomorphic graphs \cite{BabaiKucera79, Babai1980} and all trees \cite{Immerman1990}. However, it also fails on some very simple instances 
(see for example Figure \ref{fig:graphs}). To address these limitations, a hierarchy of higher-dimensional versions of this algorithm ($k$-WL for $k > 1$) was proposed, see Section \ref{chap:higher} for details.

An alternative characterization of fractional isomorphism is in terms of the notion of equitable partition. Let $G=(V;E)$ be a graph and $\partp=\{\ipartp_i \mid i\in I\}$ be a partition of $V$. $\partp$ is said to be \textit{equitable} if for every $i,j\in I$ there exists an integer $c_{i,j}$ such that for every $v\in \ipartp_i$ we have \begin{displaymath}
|N(v) \cap \ipartp_j|=c_{i,j}\end{displaymath}
where $N(v)$ denotes the neighbourhood of $v$ in $G$. The integers $c_{i,j}$ are called collectively the {\em parameters} of the partition.
Two graphs $G$, $H$ are said to have a {\em common equitable partition} if there exist equitable partitions 
$\{\ipartp_i^G \mid i\in I\}$ of $G$ and $\{\ipartp_i^H \mid i\in I\}$ and $H$ with the same parameters satisfying $|\ipartp^G_i|=|\ipartp^H_i|$ for all $i\in I$.

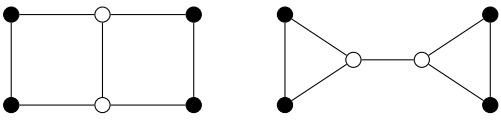
\begin{figure}[t] 
\begin{subfigure}{0.48\textwidth}
    \centering
    \graphs{3}{0.6}{0.2}
    \caption{It is easy to see that the graphs depicted above have the same iterated degree sequence. In particular, the equivalence classes are given by the distinct colours, white and black. However, the graphs are clearly not isomorphic.}  \label{fig:graphs}
  \end{subfigure}%
  
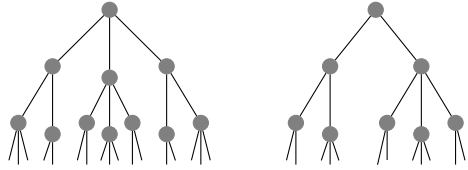
\hspace*{\fill}   
  \begin{subfigure}{0.48\textwidth}
    \centering
    \trees{3.5}{0.6}{0.2}
    \caption[Iterated degree as infinite tree.]{The first three levels of the ``local view'' of the white nodes (left) and black nodes (right) of the graphs in Figure \ref{fig:graphs}.}
    \label{fig:trees}
  \end{subfigure}%
\caption{A pair of non-isomorphic graphs that are not distinguished by 1-WL.}
\end{figure}

It turns out that all of the concepts introduced above are equivalent for the purpose of distinguishing non-isomorphic graphs. In fact, much more is true.
Let $\counting$ denote the first-order logic enriched with counting quantifiers - that is, objects of the form $\exists^n x$ for each positive integer $n$, where the meaning of a sentence  $\exists^n x \varphi(x)$ is ``there exist at least n objects $x$ that satisfy $\varphi$''; and let $\counting^k$ denote the $k$-variable fragment of $\counting$. Then, we have the following characterization of fractional isomorphism:

\begin{theorem}
\label{th:fractional-graph}
Let $G$, $H$ be graphs. The following are equivalent:
\begin{enumerate} 
\item There exists a doubly stochastic matrix $\matleft$ such that $\matleft N_{G}=N_{H}\matleft$; \label{item:1fr}
\item $G$ and $H$ have the same iterated degree sequence; \label{item:2fr}
\item $G$ and $H$ have a common equitable partition; \label{item:3fr}
\item $G$ and $H$ satisfy the same formulae in the logic $\counting^2$; \label{item:4fr}
\item $\hom(T;G)=\hom(T;H)$ for all trees $T$. \label{item:5fr}
\end{enumerate}
\end{theorem}

The equivalence of (\ref{item:1fr}), (\ref{item:2fr}) and (\ref{item:3fr}) can be traced back to the work of Leighton \cite{Leighton1982}, Tinhofer \cite{Tinhofer1986, Tinhofer1991}, and Ramana, Scheinerman and Ullman \cite{Ramana1994}.
The equivalence of (\ref{item:2fr}) and (\ref{item:4fr}) is due to Cai, F{\"{u}}rer and Immerman \cite{Cai92}, and item (\ref{item:5fr}) was added later independently by Dvor{\'{a}}k \cite{Dvorak2010} and Dell, Grohe and Rattan \cite{Dell2018} (in fact, both \cite{Cai92} and \cite{Dvorak2010,Dell2018} show much stronger results, see Section \ref{sec:kWL}).

\subsection{From graphs to relational structures} \label{sec:friso-relstr} We now extend each of the notions defined above in the context of graphs to arbitrary relational structures. We start by adapting the 1-dimensional Weisfeiler-Leman algorithm to calculate iterative refinements of a colouring of the universe and constraint set of a relational structure. While there are syntactical differences, when run on graphs this procedure is equivalent for all purposes to 1-WL.

In what follows it will be convenient to allow disconnected instances. Consider the labelled factor graph $\graph{\bfa}$ of a relational structure $\bfa$. Let $v$ be a node of $\graph{\bfa}$ and denote its neighbourhood in the factor graph by $N(v)$.
For every $j\geq 0$ and $v\in A\cup \con\bfa$, we define inductively the $j\ith$ \textit{iterated degree} $\delta^{\bfa}_j(v)$ of $v$ on $\bfa$ as follows. We set $\delta_{0}^\bfa(v)$ to be one of two arbitrary symbols that distinguish elements of $X$ from  elements of $\con\bfa$. For $j\geq 1$ we set $\delta_{j}^\bfa(v) =\{\{(\ell_{\{v,w\}}, \delta^\bfa_{j-1}(w)) \mid w \in N(v)\}\}$, and the \emph{iterated degree} of $v$ is defined accordingly as $\delta^{\bfa}(v) = (\delta_0^{\bfa}(v),\delta_1^{\bfa}(v),\delta_2^{\bfa}(v), \ldots)$. For vertices $v$ and $v'$ we write $v \eqone v'$  if $\delta^{\bfa}(v) = \delta^{\bfa}(v')$. In this case, we say that $v$ and $v'$ are \textit{$\eqone$-equivalent}.\footnote{The ``1'' refers to ``1-dimensional''.} It can be easily shown (see e.g. \cite{Butti2022distr-TOCS}) that as $j$ increases, the partition induced by $\delta^{\bfa}_{j}$ gets more refined or remains unchanged, and indeed it reaches a fixed point for some $j \leq 2|A|$.

The \textit{iterated degree sequence} of a relational structure $\bfa$ is defined as $\delta(\bfa) = \{\{\delta^\bfa(v) \mid  v \in A \cup \con\bfa \}\}$; for two $\sig$-structures $\bfa$, $\bfb$, we write $\bfa \eqone \bfb$ if they have the same iterated degree sequence.
Notice that in order to prove that $\bfa \eqone \bfb$ it is sufficient to show that $\mult{\delta^{\bfa}(a) \mid a \in A }=\mult{\delta^{\bfb}(b) \mid b \in B}$.

To lift the other equivalent characterizations of the 1-dimensional Weisfeiler-Leman algorithm from graphs to relational structures we need to define some more concepts. 
For a signature $\sig$, recall that $\LAB_\sig := \{(S,R) \mid R \in \sig, S \subseteq [\ar(R)] \}$ denotes the set of possible labels for the edges of the factor graph of a $\sig$-structure. The \textit{matrix representation} $M_{\bfa}$ of a $\sig$-structure $\bfa$ is an $A \times \con\bfa$ matrix whose entries are elements of $\LAB_\sig$, where for all $a \in A$ and $R(\ba) \in \con\bfa$, we have that $ M_{\bfa}[a,R(\ba)] = (S,R) $ where $S = \{i \in [\ar(R)] \mid  a=\ba[i]\}$.
Moreover, for every $\ell \in \LAB_\sig$, we define $M_{\bfa}^{\ell} \in \{0,1\}^{A \times \con\bfa}$  by setting $M_{\bfa}^{\ell}[a,R(\ba)]=1$ if $ M_{\bfa}[a,R(\ba)]=\ell$ and $M_{\bfa}^{\ell}[a,R(\ba)]=0$ otherwise. 


A \textit{partition} of a $\sig$-structure $\bfa$ is a pair $(\partp,\partq)$ where $\partp=\{\ipartp_i \mid i\in I\}$ is a partition of $A$ and $\partq=\{\ipartq_j \mid j\in J\}$ is a partition of $\con\bfa$. We say that
$(\partp,\partq)$ is \textit{equitable} if for every $i\in I$, $j\in J$, and $\ell\in \LAB_\sig$, there are integers $c^{\ell}_{i,j},
d^{\ell}_{j,i}$, called the {\em parameters} of the partition, such that for every $i\in I$, every $a\in \ipartp_i$, every $\ell\in \LAB_\sig$, and every $j\in J$, we have \begin{equation} \label{eq:eqpartQ} \tag{P1}
|\{R(\ba)\in \ipartq_j \mid M_{\bfa}[a,R(\ba)]=\ell\}|=c^{\ell}_{i,j}
\end{equation}
\noindent and, similarly, for every $j\in J$, every $R(\ba)\in \ipartq_j$,  every $\ell\in \LAB_\sig$, and every $i\in I$  we have \begin{equation} \label{eq:eqpartP} \tag{P2}
|\{a\in \ipartp_i \mid M_{\bfa}[a,R(\ba)]=\ell\}|=d^{\ell}_{j,i}. \end{equation} 
As above, we say that two structures $\bfa$, $\bfb$ have a {\em common equitable partition} if there are equitable partitions 
$(\{\ipartp_i^\bfa \mid i\in I\},\{\ipartq_j^\bfa \mid j\in J\})$ and $(\{\ipartp_i^\bfb \mid i\in I\},\{\ipartq_j^\bfb \mid j\in J\})$ of $\bfa$
and $\bfb$ with the same parameters satisfying $|\ipartp^\bfa_i|=|\ipartp^\bfb_i|$ for every $i\in I$ and $|\ipartq^\bfa_j|=|\ipartq^\bfb_j|$ for every $j\in J$.

Then, we have the following theorem, which generalizes the equivalent formulations of fractional isomorphism for graphs from Theorem  \ref{th:fractional-graph} to arbitrary relational structures. The proof is deferred to Section \ref{chap:higher}.

\begin{restatable}{theorem}{ledigraph}
\label{le:digraph}
Let $\bfa$, $\bfb$ be $\sig$-structures. The following are equivalent:
\begin{enumerate} 
\item There exist doubly stochastic matrices $\matleft$, $\matright$ such that for every $\ell\in \LAB_\sig$ it holds that $\matleft M_{\bfa}^{\ell}=M_{\bfb}^{\ell} \matright $ and $M_{\bfa}^{\ell}\matright^T=\matleft^TM_{\bfb}^{\ell}$; \label{1digraphItem}
\item $\bfa$ and $\bfb$ have the same iterated degree sequence; \label{2digraphItem}
\item $\bfa$ and $\bfb$ have a common equitable partition; \label{3digraphItem}
\end{enumerate}
If, additionally, $\bfa$ and $\bfb$ are graphs, then the following is also equivalent:
\begin{enumerate}
    \setcounter{enumi}{3}
    \item There exists a doubly stochastic matrix $\matleft$ such that $\matleft\na = \nb \matleft$. \label{5digraphItem}
\end{enumerate}
\end{restatable}

Then, we can write $\bfa \eqone \bfb$ if $\bfa$ and $\bfb$ satisfy any of the conditions of Theorem \ref{le:digraph}. 
Note in particular that condition (\ref{5digraphItem}) shows that our definition of $\eqone$ coincides with the standard notion of fractional isomorphism when $\bfa$ and $\bfb$ are graphs. 
Moreover, we note that \cite{ButtiD21fractional} includes an additional characterization of $\eqone$ in terms of counting homomorphisms from tree-like structures, akin to item \ref{item:5fr} in Theorem \ref{th:fractional-graph}. For simplicity, we do not include this characterization as it will not be needed in this paper.

For a $\sig$-structure $\bfa$ we say that $\cspa$ is closed under $\eqone$-equivalence if for any two $\sig$-structures $\bfx_1$ and $\bfx_2$, $\bfx_1 \eqone \bfx_2$ implies that $\bfx_1 \to \bfa$ if and only if $\bfx_2 \to \bfa$. We shall call these families of CSPs \textit{Weisfeiler-Leman invariant} to stress the connection of $\eqone$ with the Weisfeiler-Leman method.

\section{Linear Programming and the CSP} \label{sec:LP}

In this section we focus on linear programming relaxations of the homomorphism problem for general relational structures, i.e., the CSP.

\subsection{The Basic Linear Programming relaxation}

Given a pair of $\sig$-structures $\bfx$ and $\bfa$, the \emph{basic LP relaxation} (see for example \cite{Kun2012}), denoted $\blp\xa$, is defined as follows. There is a variable $p_x(a)$ for every $x \in X$ and $a \in A$, and a variable $p_{R(\bx)}(\ba)$ for every $R(\bx) \in \con{\bfx}$ and every $\ba \in A^{\ar(R)}$. All variables must take values in the range $[0,1]$. The value of $p_{x}(a)$ is interpreted as the probability that $x$ is assigned to $a$, and similarly, the value of $p_{R(\bx)}(\ba)$ is interpreted as the probability that the $\bx$ is assigned component-wise to $\ba$. The variables are restricted by the following equations:
\begin{align}
& \sum_{a \in A} p_x(a)=1  &&  x \in X  \label{eq:BLP1thesis} \tag{\blp1}\\
&  p_x(a) = \sum_{\mathclap{\ba \in A^{\ar(R)},\ a_i=a}} \, p_{R(\bx)}(\ba) &&  a\in A, \ R(\bx) \in \con{\bfx}, \ i \in [\ar(R)] \mbox{ such that } x_i=x  \label{eq:BLP2thesis} \tag{\blp2} \\
& p_{R(\bx)}(\ba) = 0 &&  R(\bx) \in \con{\bfx}, \  \ba \in A^{\ar(R)} \textnormal{ such that } R(\ba) \notin \con\bfa \label{eq:BLP3thesis} \tag{\blp3}
\end{align}
Intuitively, equations (\ref{eq:BLP1thesis})--(\ref{eq:BLP3thesis}) ensure that, for each $R(\bx) \in \con{\bfx}$, the values of $p_{R(\bx)}(\ba)$ form a probability distribution on $\RA$ (which is additionally consistent with $p_x(a)$'s).

\subsection{The Sherali–Adams hierarchy for CSPs} 

We use the Sherali-Adams method to strengthen the BLP and obtain a hierarchy of increasingly tighter relaxations of the CSP. In this presentation we will follow \cite{Atserias2013}.

Let $\polytope \subseteq[0,1]^n$ be a polytope $\{\by\in \mathbb{R}^n : M \by\geq \bb, 0\leq\by\leq 1\}$ for a matrix $M\in\mathbb{R}^{m\times n}$, and a column vector $\bb\in \mathbb{R}^m$. We denote the convex hull of the $\{0,1\}$-vectors in $\polytope$ by $\polytope^{\mathbb{Z}}$. The sequence of Sherali-Adams relaxations of $\polytope^{\mathbb{Z}}$ is the sequence of polytopes $\polytope=\polytope^1\supseteq \polytope^2\supseteq\cdots$ where $\polytope^k$ is defined in the following way.

Each inequality in $M \by\geq\bb$ is multiplied by all possible terms of the form $\Pi_{i\in I} y_i\Pi_{j\in J} (1-y_j)$ where $I,J\subseteq[n]$ satisfy $|I\cup J|\leq k-1$ and $I\cap J=\emptyset$. This leaves a system of polynomial inequalities, each of degree at most $k$. Then, this system is {\em linearized} and hence relaxed in the following way: each square $y_i^2$ is replaced by $y_i$ and each resulting monomial $\Pi_{i\in K} y_i$ is replaced by a variable $z_K$. In this way we obtain a polytope $\polytope^k_L$. Finally, $\polytope^k_L$ is projected back to $n$ dimensions by defining 
\begin{displaymath} \polytope^k:=\{\by\in \mathbb{R}^n: \text{there exists $\bz\in \polytope^k_L$ such that $\bz_{\{i\}}=\by_i$ for each $i\in [n]$}\}.\end{displaymath}
We note here that $\polytope^{\mathbb{Z}}\subseteq \polytope^k$ for every $k\geq 1$.

In order to apply the SA method to the homomorphism problem there are different possible choices for the polytope $\polytope$ (encoding a relaxation of homomorphism) to start with, each one then yielding a different hierarchy. 
Here we shall adapt a SA-based family of relaxations commonly used in optimization variants of CSP \cite{GeorgiouMT09,ChanLRS13,YoshidaZ14,ThapperZ17,Ciardo2022Tensors} which we transform into a relaxation of (plain) CSP by just turning the objective function into a set of new restrictions. Hence, the resulting system of inequalities is not, strictly speaking, obtained using the SA method. Nonetheless, we shall abuse slightly notation and still use $\sa^k$ to refer to our system of inequalities. 

In fact, giving an explicit description of the inequalities obtained using the SA method for any natural polytope $\polytope$ encoding the LP relaxation for a {\em general} CSP in our setting is a bit cumbersome (because the constraints of the CSP are encoded in the polytope-defining inequalities rather than in the objective function as in the optimization variants).\footnote{A precise description of these can be found in \cite{atserias2022}.} Hence, it seems sensible to settle for a good approximation as $\sa^k$. Indeed, as we discuss in Appendix \ref{app:realSAmethod} (see Lemma \ref{le:SAkrPk}), the sequence of relaxations $\sa^k$ is tightly interleaved with the sequence $\polytope^k$ obtained by the SA method, in stricto sensu, for a natural choice of initial polytope $\polytope$.

Given a pair of $\sig$-structures $\bfx$ and $\bfa$, the system of inequalities $\sak$ for the homomorphism problem over $\xa$ contains a variable $p_V(f)$ for every $V \subseteq X$ with $1 \leq |V| \leq k$ and every $f:V \to A$, and a variable $p_{R(\bx)}(f)$ for every $R(\bx) \in \con{\bfa}$ and every $f:\{\bx\} \to A$. Each variable must take a value in the range $[0,1]$. The variables are constrained by the following conditions:
\begin{align}
&\sum_{f:V \to A} p_V(f)=1  &&  V\subseteq X \textnormal{ s.t. } |V|\leq k  \label{eq:SA1} \tag{\sa1} \\
&p_U (f) = \sum_{\mathclap{g:V \to A,g|_{U} =f}} \, p_V(g) && U \subseteq V \subseteq X \, \textnormal{ s.t. }|V|\leq k, \ f:U \to A \label{eq:SA2} \tag{\sa2} \\
&p_U(f) = \sum_{\mathclap{g:V\to A, g|_U =f}} \, p_{R(\bx)}(g) &&  R(\bx)\in{\mathcal C}_{\bfx}, \ U\subseteq\{\bx\}=V \textnormal{ s.t. } |U| \leq k, f: U \to A \label{eq:SA3}\tag{\sa3} \\
&p_{R(\bx)}(f) = 0 &&  R(\bx) \in \con{\bfx}, \ f : \{\bx\} \to A \textnormal{ s.t. }f(\bx) \not \in \RA\label{eq:SA4}\tag{\sa4}
\end{align} 

For the particular case of $k=1$ we shall use the simplified notation $p_x(f(x))$ to denote the variable $p_V(f)$ for a singleton set $V=\{x\}$ and a function $f:V \to A$.

It is easy to see that $\sa^1\xa$ is equivalent to $\blp\xa$, subject to the following additional constraint: \begin{align}
  & p_{R(\bx)}(\ba) = 0  \qquad  &&  R(\bx) \in \con{\bfx}, \ \ba \in A^{\ar(R)} \colon \exists i,j \in [\ar(R)] \label{eq:SA4loops} \tag{$\circlearrowleft$}\\
  &&& \textnormal{ such that } x_i=x_j  \textnormal{ and }a_i \neq a_j \nonumber
\end{align}
That is, $\blp$ and $\sa^1$ differ only in the way that they deal with repeated entries in a tuple.\footnote{We remark that in the literature the difference between the two relaxations is sometimes neglected, which occasionally leads to unjustified or slightly incorrect claims.}

For a linear program $\textnormal{L} \in \{\blp,\sa^1\}$ we say that $\textnormal{L}\xa$ is \textit{feasible} if there exists a rational solution to the system $\textnormal{L}\xa$. We say that $\textnormal{L}$ \textit{decides} $\cspa$ if, for every input structure $\bfx$, if $\textnormal{L}\xa$ is feasible then $\bfx \to \bfa$. 

While in some specific instances with repeated entries $\sa^1$ is strictly stronger than $\blp$, in general these two relaxations have the same expressive power in terms of solving fixed-template CSPs. In particular, $\blp$ (equivalently, $\sa^1$) solves $\cspa$ iff $\bfa$ has symmetric polymorphisms of all arities (see e.g. \cite{Kun2012}).

The power of the Sherali-Adams method for the CSP is also well-understood. In particular, the third level of the SA hierarchy already solves all CSPs of bounded width (this follows from the more general results of \cite{ThapperZ17} together with the Barto-Kozik theorem \cite{Barto2014local} and the collapse of the bounded width hierarchy \cite{Barto2016Collapse}).

\section{Weisfeiler-Leman Invariant CSPs} \label{chap:sawl}

\subsection{The decomposition theorem} \label{sec:decomp}

We now present the decomposition theorem for $\sa^1$. Note that condition (\ref{2SA1Item}) and, for graphs, condition (\ref{5SA1Item}) below are naturally seen as the homomorphism counterpart of the notion of fractional isomorphism (see condition (\ref{1digraphItem}) in Theorem \ref{le:digraph}). 

A version of this theorem for valued structures will follow in Section \ref{chap:pvcsp} and a higher-dimensional version in Section \ref{chap:higher}.

\begin{theorem}\label{thm:SA1}
Let $\bfx$, $\bfa$ be $\sig$-structures. Then, the following are equivalent:
\begin{enumerate}
    \item $\sa^1(\bfx,\bfa)$ is feasible; \label{1SA1Item}
    \item  There exist left stochastic matrices $\matleft$, $\matright$ such that for every $\ell=(S,R) \in \LAB_\sig$ it holds that
    $\matleft M_{\bfx}^{\ell}\leq \sum_{\ell'} M_{\bfa}^{\ell'}\matright$, where $\ell'$ ranges over all $(S',R)\in \LAB_\sig$ with $S\subseteq S'$; \label{2SA1Item}
    \item There exists a sequence of structures $\bfy_0, \ldots, \bfy_n$ such that $\bfy_0=\bfx$, $\bfy_n=\bfa$, and for all $i=0,\ldots,n-1$ we have that $\bfy_i \to \bfy_{i+1}$ or $\bfy_{i} \eqwl[1] \bfy_{i+1}$; \label{3SA1Item}
    \item There exists a pair of structures $\yone, \ytwo$ such that $\bfx \to \bfy_{1}$,  $\bfy_{1} \eqwl[1] \bfy_{2}$, and $\ytwo \to \bfa$. \label{4SA1Item}
\end{enumerate}
If in addition none of the relations in $\bfx$ and $\bfa$ have repetitions, then the following condition is also equivalent:
\begin{enumerate}
    \setcounter{enumi}{4}
    \item There exist left stochastic matrices $\matleft$, $\matright$ such that for every $\ell \in \LAB_\sig$ it holds that 
    $\matleft M_{\bfx}^{\ell}=M_{\bfa}^{\ell}\matright$.\label{5SA1Item}
\end{enumerate}
\end{theorem}

\begin{proof} The equivalence $(\ref{1SA1Item})\Leftrightarrow(\ref{2SA1Item})$ is merely syntactic. In particular we shall show that there is a one-to-one satisfiability-preserving correspondence between pairs of matrices and variable assignments of $\saxa$. However, we first need to massage a bit the two formulations. First, we can assume that for every $R(\bx)\in {\mathcal C}_{\bfx}$ and $R'(\ba) \in {\mathcal C}_{\bfa}$, the corresponding entry in $\matright$ is null unless $R=R'$ and $f(\bx)=\ba$ for some  $f:\{\bx\}\rightarrow\{\ba\}$, since otherwise it is not possible for $\matright$ to form part of a feasible solution. Secondly, we note that the feasibility of $\saxa$ does not change if in (\ref{eq:SA3}) we replace
$=$ by $\leq$ obtaining a new set of inequalities (which to avoid confusion we shall denote by (\hypertarget{eq:SA3prime}{$\sa3'$})) and, in addition, we add
for every $R(\bx) \in \mathcal{C}_\bfx$ the equality
\begin{equation} \label{eq:SA5}\tag{\sa5}
\sum_{f:\{\bx\}\rightarrow A} p_{R(\bx)}(f)=1.
\end{equation}
Finally, note that in $\saxa$ we can ignore (\ref{eq:SA2}).

Then we can establish the following correspondence between pairs of matrices $\matleft$, $\matright$ and assignments $\saxa$:
for every $x\in X$ and $a \in A$, we set $p_x(a)=\matleft[a,x]$ and for every $R(\bx) \in \mathcal{C}_\bfx$ and $f:\{\bx\} \to A$ we define $p_{R(\bx)}(f)=\matright[R(f(\bx)),R(\bx)]$. Then, it is easy to see that (\hyperlink{eq:SA3prime}{$\sa3'$}) corresponds to $\matleft M_{\bfx}^{\ell}\leq \sum_{\ell' \in \LAB_\ell} M_{\bfa}^{\ell'}\matright$ for every $\ell\in \LAB_\sig$ (where $\LAB_{(S,R)}:=\{(S',R) \in \LAB_\sig \mid S \subseteq S'\}$), $\matleft$ being left stochastic corresponds to (\ref{eq:SA1}), and $\matright$ being left stochastic corresponds to (\ref{eq:SA5}).

The equivalence $(\ref{1SA1Item}) \Leftrightarrow (\ref{5SA1Item})$ is obtained as in $(\ref{1SA1Item}) \Leftrightarrow (\ref{2SA1Item})$. We just need to notice that when $\bfx$ and $\bfa$ have no loops, then none of the entries in $M_{\bfx}$ and $M_{\bfa}$ contain any label $\ell=(S,R)\in \LAB_\sig$ where $|S|>1$ and hence it is only necessary to consider labels $\ell=(S,R)\in \LAB_\sig$
where $S$ is a singleton. Observe that, in this case, the equation in $(\ref{2SA1Item})$ becomes 
$\matleft M_{\bfx}^{\ell}\leq M_{\bfa}^{\ell}\matright$ since for every label $\ell=(S,R)$ where $S$ is a singleton, the only label $(S',R)$ with $S\subseteq S'$ and $M_{\bfa}^{\ell}$ not a zero matrix is $\ell$ itself. Finally, in order to 
replace $\leq$ by $=$ in the previous equation we just need to use (\ref{eq:SA3}) instead of (\hyperlink{eq:SA3prime}{$\sa3'$}).

Notice that $(\ref{4SA1Item}) \Rightarrow (\ref{3SA1Item})$ is trivial.

The proof of $(\ref{3SA1Item}) \Rightarrow (\ref{2SA1Item})$ is by induction on $n$. If $n=0$ the claim is immediate, so assume that $n \geq 1$. Let $\bfy_0, \yone, \ldots, \bfy_n$ be a sequence of structures satisfying (\ref{3SA1Item}). By the induction hypothesis, there exist left stochastic matrices $\matleft$, $\matright$ such that $\matleft M_{\yone}^{\ell}\leq \sum_{\ell' \in \LAB_\ell} M_{\bfy_n}^{\ell'}\matright$ for all $\ell\in \LAB_\sig$.

If $\bfy_0 \eqwl[1] \yone$ then it follows from Theorem \ref{le:digraph} that there exist doubly stochastic matrices $\matleft'$ and $\matright'$ such that $\matleft'M_{\bfy_0}^{\ell}= M_{\yone}^{\ell}\matright'$ for all $\ell \in \LAB_\sig$, and so it is easy to verify that $\matleft\matleft'$, $\matright\matright'$ are such that (\ref{2SA1Item}) holds. Assume that $\bfy_0 \to \yone$. We shall show that there exist left stochastic matrices $\matleft'$ and $\matright'$ such that for all $y \in Y_1$, for all $R(\bx) \in \con\bfx$, and for all $\ell \in \LAB_\sig$ there exists $\hat{\ell}=\hat{\ell}(y,R(\bx),\ell) \in \LAB_\ell$ such that $\matleft' M_{\bfx}^{\ell}[y,R(\bx)] \leq M_{\bfy_1}^{\hat{\ell}}\matright'[y,R(\bx)]$. Assuming that this holds, again it follows by the induction hypothesis that $\matleft\matleft'$, $\matright\matright'$ are left stochastic matrices such that $\matleft\matleft'M_{\bfx}^\ell \leq \sum_{\ell' \in \LAB_\ell} M_{\bfa}^{\ell'} \matright \matright' $ for all $ \ell \in \LAB_\sig$.

Let $h$ be a homomorphism from $\bfx$ to $\bfy_1$. We define $\matleft'[y,x]=1$ if $y=h(x)$ and $\matleft'[y,x]=0$ otherwise. Similarly, we set $\matright'[R'(\by),R(\bx)]=1$ if $\by=h(\bx)$ and $R'=R$ and $\matright'[R'(\by),R(\bx)]=0$ otherwise. It is easy to see that $\matleft'$ and $\matright'$ are left stochastic. Now let $\ell =(S,R) \in \LAB_\sig$, $y \in Y_1$ and $R(\bx) \in \mathcal{C}_\bfx$. If $M^\ell_\bfx[x,R(\bx)]=0$ for all $x \in X$ then $\matleft' M^\ell_\bfx[y,R(\bx)]=0$ and there is nothing to prove. So we can assume that there is $x \in X$ such that for all $i \in [\ar(R)]$, $\bx[i] =x$ if and only if $i \in S$. Then we have that $\matleft' M^\ell_\bfx[y,R(\bx)]=1$ if $y=h(x)$, and $\matleft' M^\ell_\bfx[y,R(\bx)]=0$ otherwise. Again in the latter case there is nothing to prove so let us assume that $y=h(x)$. It follows that $h(\bx)[i]=y$ for all $i \in S$ and hence there exists $\hat{\ell}=(R,S')$ with $S \subseteq S'$ such that $M^{\hat{\ell}}_{\bfy_1}[y,R(h(\bx))]=1$, which completes the proof.

For $(\ref{1SA1Item}) \Rightarrow (\ref{4SA1Item})$, assume that $\saxa$ is feasible.
Let $p_x(a)$, $p_{R(\bx)}(\ba)$ form a feasible solution of $\saxa$ and let $m>0$ be an integer such that all the values $mp_x(a)$ and $mp_{R(\bx)}(\ba)$ are (non-negative) integers.

We define the universe of both structures $\yone$ and $\ytwo$ as $Y_1=Y_2 = [m] \times X$. The structure $\yone$ is simply a disjoint union of $m$ copies of $\bfx$: for every $R \in \sigma$, $\bx \in X^{\ar(R)}$ and $k \in [m]$ we set
\begin{displaymath}
((k,\bx[1]), (k,\bx[2]), \dots, (k,\bx[\ar(R)])) \in R^{\yone} \Leftrightarrow \bx \in R^{\bfx}.\end{displaymath}
Observe that clearly $\bfx \to \yone$, since for all $k \in [m]$ the map $h_k: x \mapsto (k,x)$ is a homomorphism. Also notice that for all $k \in [m]$, the iterated degree $\delta^{\yone}(k,x)$ is the same as $\delta^\bfx(x)$.

The structure $\ytwo$ is a ``twisted'' version of $\yone$ (the construction is a version of the twisted product  from~\cite{Kun13}).
For every $x \in X$, fix a tuple $\bp_x \in A^m$ in which $a \in A$ appears exactly $mp_x(a)$ times -- note that this is possible since the $mp_x(a)$ sum up to $m$ by (\ref{eq:SA1}). Moreover, for every $R(\bx) \in \mathcal{C}_\bfx$, denote $r=\ar(R)$, and consider an $m \times r$ matrix $T=T(R(\bx))$ that has, for each $\ba \in A^{r}$, exactly $mp_{R(\bx)}(\ba)$ rows equal to $\ba$. Note that then all the rows of $T$ are elements of $R^\bfa$ by (\ref{eq:SA4}), and that the $i\ith$ column of $T$ contains $a \in A$ exactly $mp_{\bx[i]}(a)$ times by (\ref{eq:SA3}), in other words, the multiset of elements of this columns is equal to $\mult{\bp_{\bx[i]}}$. In particular, $T$ indeed has $m$ rows. Moreover, if $\bx[i] = \bx[j]$, then the $i\ith$ and $j\ith$ columns of $T$ are identical by (\ref{eq:SA4loops}). It follows that there are permutations $\rho_1, \dots, \rho_r$ over $[m]$ such that 
\begin{enumerate}
  \item for every $k \in [m]$,
  $(\bp_{\bx[1]}[\rho_1(k)], \bp_{\bx[2]}[\rho_2(k)], \dots, \bp_{\bx[r]}[\rho_r(k)])$ is equal to the $k\ith$ row of $T$; \label{property:1}
  \item for every $i,j \in [r]$, if $\bx[i]=\bx[j]$ then $\rho_i=\rho_j$. \label{property:2}
\end{enumerate}
Then, for every $R \in \sigma$, $\bx \in X^{\ar(R)}$ and $k \in [m]$ we set \begin{displaymath}
((\rho_1(k),\bx[1]), (\rho_2(k),\bx[2]), \dots, (\rho_{r}(k),\bx[r])) \in R^{\ytwo} \Leftrightarrow \bx \in R^{\bfx}.\end{displaymath}

We define $h: Y_2 \to A$ by $h(k,x) = \bp_x[k]$ for all $k \in [m]$ and $x \in X$. It is easy to see that the image of any tuple in $R^{\ytwo}$ under $h$ is a row of $T(R(\bx))$ for some $R(\bx) \in \con\bfx$, and hence belongs to $R^\bfa$. In particular, for every $\by \in R^{\ytwo}$, there exists some $\bx \in R^\bfx$ and $k \in [m]$ such that $\by[i] = (\rho_i(k),\bx[i])$ for all $i \in [\ar(R)]$. Then, $h(\by[i])=\bp_{\bx[i]}(\rho_i(k))$ and hence $h(\by)$ is the $k\ith$ row of $T(R(\bx))$ by (\ref{property:1}).

Moreover, for all $(k,x) \in Y_2$ the iterated degree $\delta^{\ytwo}(k,x)$ is the same as $\delta^\bfx(x)$ and hence the same as in $\yone$ (note here that item (\ref{property:2}) above guarantees that repeated entries are handled correctly). It follows that $\yone \eqone \ytwo$, and the proof is concluded. 
\end{proof}

\subsection{Characterizing Weisfeiler-Leman invariant CSPs}

The principal novelty of our result is that it gives us an alternative combinatorial description of solvability by the first level of the Sherali-Adams relaxation and thus an improved understanding of solvability of CSPs by linear programs. A concrete application is the answer to the following question: for which structures $\bfa$ is $\csp(\bfa)$ closed under $\eqwl[1]$-equivalence? This question, which also arises in the context of distributed CSPs (see \cite{Butti2022distr-TOCS}), finds an answer in the following characterization of Weisfeiler-Leman invariant CSPs:

\begin{theorem} \label{th:main-csp}
Let $\bfa$ be a fixed finite $\sig$-structure. The following are equivalent:
\begin{enumerate}
    \item $\csp(\bfa)$ is closed under $\eqwl[1]$-equivalence; \label{2-main-csp}
    \item $\sa^1$ decides $\csp(\bfa)$; \label{3-main-csp}
    \item $\blp$ decides $\csp(\bfa)$; \label{4-main-csp}
    \item $\bfa$ has symmetric polymorphisms of all arities. \label{5-main-csp}
\end{enumerate}
\end{theorem}

\begin{proof}

The implication $(\ref{2-main-csp}) \Rightarrow (\ref{3-main-csp})$ is an immediate corollary of item (\ref{4SA1Item}) in Theorem \ref{thm:SA1}.

For $(\ref{3-main-csp}) \Rightarrow (\ref{2-main-csp})$, it is enough to notice that item (\ref{1digraphItem}) in Theorem \ref{le:digraph} is a stronger condition than item (\ref{2SA1Item}) in Theorem \ref{thm:SA1} and therefore, for any two similar structures $\bfx, \bfx'$, if $\bfx \eqone \bfx'$ then $\sa^1(\bfx,\bfx')$ is feasible. Alternatively, this implication is also a direct consequence of the equivalence of (\ref{1SA1Item}) and (\ref{3SA1Item}) in Theorem \ref{thm:SA1}.

For $(\ref{3-main-csp}) \Rightarrow (\ref{4-main-csp})$ we will use the following fact which is a consequence of the Sparse Incomparability Lemma \cite{Nesestril1989}: for every $\sig$-structure $\bfx$, there exists a structure $\bfx'$ with no loops such that $\bfx'\rightarrow \bfx$ and $\bfx\rightarrow \bfa$ iff $\bfx'\rightarrow \bfa$. Now, assume that $\blp$ does not decide $\csp(\bfa)$. This means that there exists a structure $\bfx$ not homomorphic to $\bfa$ and such that $\blp\xa$ is feasible. Let $\bfx'$ be the structure given by the Sparse Incomparability Lemma. Since $\bfx'\rightarrow \bfx$ it follows that $\blp(\bfx',\bfa)$ is feasible, and, since $\bfx'$ has no loops, $\sa^1(\bfx',\bfa)$ is feasible as well. Since $\bfx'$ is not homomorphic to $\bfa$ it follows that $\sa^1$ does not decide $\csp(\bfa)$.

The implication $(\ref{4-main-csp}) \Rightarrow (\ref{3-main-csp})$ follows immediately by comparing the inequalities of $\sa^1$ and $\blp$.

Finally, the equivalence of (\ref{4-main-csp}) and (\ref{5-main-csp}) is well-known, see for instance \cite{Kun2012}. 
\end{proof}

\section{Weisfeiler-Leman Invariant Promise Valued CSPs} \label{chap:pvcsp}

\subsection{Promise Valued CSPs} \label{sec:prelim-pvcsp}

\paragraph{PVCSP.} Let $\mathbb{Q}$ be the set of rational numbers.
We denote by $\qq$ the set of non-negative rationals and by $\qinfty$ the set $\mathbb{Q} \cup \{\infty\}$, where $\infty$ is an additional symbol interpreted as a positive infinity. We set $0 \cdot \infty = 0$ and $c \cdot  \infty = \infty$ for all $c>0$.

A $k$-ary \textit{valued relation} on $A$ is a function $R:A^k \to \qinfty$. A \textit{valued $\sig$-structure} $\bfa$ consists of a finite universe $A$, together with a valued relation $\RA$ of arity $\ar(R)$ on $A$ for each $R \in \sig$.
Valued structures are sometimes referred to as \textit{general-valued} in the literature \cite{KolmogorovTZ15,ThapperZ17} to emphasize that relations in $\bfa$ may take non-finite values. A $\sig$-structure $\bfa$ is said to be \textit{non-negative finite-valued} if for every $R \in \sig$, the range of $\RA$ is contained in $\qq$.

Let $\bfx$, $\bfa$ be valued $\sig$-structures, where $\bfx$ is non-negative finite-valued. The \emph{value} of a map $h: X \to A$ for $(\bfx,\bfa)$, and the \emph{optimum value} for $(\bfx, \bfa)$ are given by
\begin{displaymath}
\val(\bfx,\bfa,h) = \sum_{R \in \sig} \sum_{\textbf{x} \in X^{\ar(R)}} R^\bfx(\textbf{x}) R^\bfa(h(\textbf{x})), \ \quad \
\opt(\bfx,\bfa) = \min_{h:X \to A}\val(\bfx,\bfa,h).
\end{displaymath}
For two valued $\sig$-structures $\bfa$ and $\bfb$, the \textit{Promise Valued CSP over $\ab$}~\cite{ViolaZ21,AAA24}, denoted $\pvcsp \ab$, is defined as follows: given a pair $(\bfx,\thresh)$, where $\bfx$ is a non-negative finite-valued $\sig$-structure  and $\thresh \in \mathbb{Q}$ is a \emph{threshold}, output $\yes$ if $\opt\xa \leq \thresh$, and output $\no$ if $\opt\xb > \thresh$. 
 We call $(\bfa,\bfb)$ a \emph{PVCSP template} if the sets of $\yes$ and $\no$ instances are disjoint. We show in Proposition~\ref{prop:FH} that this least restrictive meaningful requirement on a PVCSP template coincides with the choice taken in~\cite{ViolaZ21}.
 
 Notice that the values of $R$ have a different intended meaning in the template valued structures $\bfa$, $\bfb$ and in the input valued structure $\bfx$. For the template, $\RA(\ba)$ and $\RB(\bb)$ should be understood as the \emph{cost} of $\ba$ and $\bb$: we wish an assignment $h$ to map tuples of variables to tuples of domain elements that are as cheap as possible (and, in fact, $\RA$ or $\RB$ is often referred to as a \emph{cost function}). On the other hand, $R^{\bfx}(\bx)$ is the \emph{weight} of tuple $\bx$: we need to be more concerned about heavy tuples, while we may ignore the tuples of zero weight. 
 For instance, the PCSP over a pair of structures $(\bfa',\bfb')$ is essentially the same problem as the PVCSP over the pair of $\{0,\infty\}$-valued structures $\ab$, where tuples in the latter template are given zero cost iff they belong to the corresponding relations in the former template; while to an instance $\bfx'$ of the PCSP corresponds a non-negative finite-valued structure $\bfx$ where the cost of a tuple is zero iff the tuple does \emph{not} belong to the corresponding relation in $\bfx'$ (and costs of the remaining tuples are arbitrary positive rationals), together with any threshold $\thresh \in \qq$.
 
 For a PVCSP input valued $\sig$-structure $\bfx$ we define the set of \emph{constraints} $\mathcal{C}_\bfx$ as the set of formal expressions of the form $R(\bx)$ where $R \in \sig$, $\bx \in X^{\ar(R)}$, and $R^{\bfx}(\bx) > 0$; the value $R^{\bfx}(\bx)$ is the \emph{weight} of the constraint.
 
 We say that a valued relation $R^\bfx$ has no repetitions if $R^\bfx(\bx)=0$ whenever $\bx$ has a repetition. Similarly, we say that an input valued structure $\bfx$ has no repetitions if none of its valued relations has a repetition.

\begin{example}
As mentioned in the introduction, the PVCSP framework can be used to model a decision version of constant-factor approximation problems for MaxCSP. More concretely, suppose that we want to find a $c$-approximation for $\csp(\bfa)$ for some (non-valued) $\sigma$-structure $\bfa$ and some $c < 1$. One can model this problem as $\pvcsp(\bfa',\bfb')$ where $A'=B'=A$ and for all $R \in \sigma$ and $\ba \in A^{\ar(R)}$,  $R^{\bfa'}(\ba)=-1$ if $\ba \in \RA$ and $R^{\bfa'}(\ba)=0$ otherwise; and $R^{\bfb'}(\ba)=\frac{1}{c} R^{\bfa'}(\ba)$. 
Given an instance $\bfx$ of $\csp(\bfa)$ and a parameter $0 < \beta \leq 1$, we turn it into an instance $(\bfx',-\beta m)$ of $\pvcsp(\bfa',\bfb')$ in a natural way, where $\bfx'$ is a 0-1 valued structure and $m$ is the number of constraints in $\bfx'$. Then, $\opt(\bfx',\bfa')\leq -\beta m$ if a $\beta$-fraction of all constraints of $\bfx$ can be satisfied in $\bfa$, and $\opt(\bfx',\bfb')>- \beta m$ if not even a $c\beta$-fraction of the constraints of $\bfx$ can be satisfied in $\bfa$.
\end{example}

\paragraph{Iterated degree.}
As in the crisp setting, we use the factor graph representation of valued structures to define the notion of iterated degree and $\eqone$-equivalence. This is defined analogously to the crisp setting, except that each edge $\{x,R(\bx)\}$ in the factor graph of $\bfx$ is now labelled by a triple $(S,R,q)$ where $q=R^\bfx(\bx)$ is the value weight of the constraint $R(\bx) \in \con\bfx$ and, as above, $S = \{i \in [\ar(R)] \mid x_i = x\}$. The notion of iterated degree and the equivalence relation $\eqone$ are then defined using the factor graph for valued structures in the same way as for crisp structures.

Moreover, recall the definition of Weisfeiler-Leman invariance from Section \ref{sec:1WL}. This translates to the PVCSP setting as follows: for any promise valued template $\ab$, $\pvcsp\ab$ is closed under $\eqone$-equivalence iff for any two non-negative finite-valued structures $\bfx_1$ and $\bfx_2$ similar to $\bfa$ and $\bfb$, $\bfx_1 \eqone \bfx_2$ implies that $\opt(\bfx_2,\bfb) \leq \opt(\bfx_1,\bfa)$.  Correspondingly, we call these templates \textit{Weisfeiler-Leman invariant PVCSPs}.

\paragraph{Linear programming relaxations.} 
Given two valued $\sig$-structures $\bfx$ and $\bfa$ where $\bfx$ is non-negative finite-valued, the systems of inequalities $\blp\xa$ and $\sa^1\xa$ are given by adapting equations (\ref{eq:BLP1thesis}), (\ref{eq:BLP2thesis}), (\ref{eq:BLP3thesis}) and (in the case of $\sa^1\xa$) (\ref{eq:SA4loops}) to the valued case, and importantly, adding an objective function. Concretely, $\blp\xa$ for valued $\bfx$ and $\bfa$ is the following linear program.
\begin{align}
& \mathrlap{\opt^{\blp} \xa := \min \sum_{R(\bx) \in \con{\bfx}}  \sum_{\ba \in A^{\ar(R)}} p_{R(\bx)}(\ba) R^\bfx(\bx) R^\bfa(\ba)} \label{eq:objBLP} \tag{$\star$}\\
\textnormal{subject} & \textnormal{ to:} \nonumber && \\[5pt]
& \sum_{a \in A} p_x(a)=1  &&  x \in X \label{eq:SA1pvcsp} \tag{v\blp1}\\
&  p_x(a) = \sum_{\mathclap{\ba \in A^{\ar(R)},\ a_i=a}} \, p_{R(\bx)}(\ba) &&  a\in A, R(\bx) \in \con{\bfx}, \ i \in [\ar(R)] \mbox{ such that } x_i=x \label{eq:SA2pvcsp} \tag{v\blp2} \\
&  p_{R(\bx)}(\ba) = 0 &&  R(\bx) \in \con{\bfx}, \  \ba \in A^{\ar(R)} \textnormal{ such that } R^\bfa(\ba) = \infty \nonumber\label{eq:SA3pvcsp} \tag{v\blp3}
\end{align}

As for $\saxa$, the objective function -- denoted $\opt^{\sa^1}\xa$ -- is the same as in $\blp\xa$. The variables are subject to all the constraints in $\blp\xa$ and, in addition, they are also subject to the following constraint which (as in the non-valued case) handles the repetitions in the constraints of $\bfx$.
\begin{align}
  & p_{R(\bx)}(\ba) = 0  && \quad  R(\bx) \in \con{\bfx}, \ \ba \in A^{\ar(R)} \colon \exists i,j \in [\ar(R)] \label{eq:SA4pvcsp} \tag{v$\circlearrowleft$} \\
&&& \quad \textnormal{such that } x_i = x_j \textnormal{ and }a_i \neq a_j \nonumber
\end{align} 
Notice that in general $\opt^\blp\xa \leq \opt^{\sa^1}\xa$. Additionally, in the particular case where $\bfx$ has no repetitions, $\blp$ and $\sa^1$ coincide and so $\opt^\blp\xa = \opt^{\sa^1}\xa$.

Moreover, for $\textnormal{L} \in \{\blp,\sa^1\}$, if there exists a rational solution to $\textnormal{L}\xa$ then $\opt^{\textnormal{L}}\xa < \infty$, since $\RA(\ba)=\infty$ implies $p_{R(\bx)}(\ba)=0$ for all $R(\bx)$ and $0\cdot\infty = 0$ (formally, one should skip these summands in (\ref{eq:objBLP})).
If $\textnormal{L}$ is infeasible, then we set $\opt^{\textnormal{L}}\xa = \infty$.

The inner sum in (\ref{eq:objBLP}) is equal to the expected ``cost'' of the constraint $R(\bx)$ with weight $R^{\bfx}(\bx)$ when $\bx$ is evaluated according to this distribution. From this observation it is apparent that $\opt^{\textnormal{L}}\xa \leq \opt \xa$. 
We say that $\textnormal{L}$ \textit{decides} $\pvcsp \ab$ if, for every input structure $\bfx$, we have $\opt \xb \leq \opt^{\textnormal{L}}\xa$. Note that in this case the algorithm for $\pvcsp \ab$ that answers $\yes$ iff $\opt^{\textnormal{L}}\xa \leq \thresh$ (where $\thresh$ is the input threshold) is correct, so the definition makes sense.%

\subsection{Fractional operations} \label{sec:fr-ops}

In this section we introduce the notions of fractional operation, which played a major role in the development of the algebraic approach to Valued CSPs \cite{Cohen2013,Kozik2015}, and discuss some of their basic properties.

An $n$-ary \textit{fractional polymorphism} of a pair of valued $\sig$-structures $\ab$ is a probability distribution $\omega$ on the set $B^{A^n}:=\{f:A^n \to B\}$ such that for every $R \in \sig$ and every list of $n$ tuples  $\textbf{a}_1,\ldots,\textbf{a}_n \in A^{\ar(R)}$ we have that
\begin{displaymath} \sum_{f \in B^{A^n}} \omega(f) R^\bfb(f(\textbf{a}_1,\ldots,\textbf{a}_n)) \leq \frac{1}{n}\sum_{i=1}^n  R^\bfa(\textbf{a}_i)\end{displaymath}
where $f$ is applied to $\textbf{a}_1,\ldots,\textbf{a}_n \in A^{\ar(R)}$ component-wise.\footnote{We use a simpler concept than fractional polymorphism as defined in, e.g., \cite{ViolaZ21}, which will be sufficient for our purposes.}
The \textit{support} of $\omega$ is the set of functions $f:A^n \to B$ such that $\omega(f) >0$. We say that $\omega$ is \textit{symmetric} if every operation in its support if symmetric.

The equivalence of solvability by $\blp$ to invariance under symmetric operations lifts to the promise valued setting.

\begin{theorem}[\cite{ViolaZ21}]  \label{th:LP-pvcsp}
Let $\ab$ be a promise valued template of signature $\sig$. The following are equivalent.
\begin{enumerate}
    \item $\blp$ decides $\pvcsp\ab$; \label{item:4thmain}
    \item $\ab$ has symmetric fractional polymorphisms of every arity. \label{item:5thmain}    
\end{enumerate}
\end{theorem}

\begin{example} \label{ex:pvcsp}
A CSP that can be decided by $\blp$ is e.g. the Horn-3-Sat, where the template has domain $\{\mathrm{true}, \mathrm{false}\}$ and two ternary relations defined by $\neg x \vee \neg y \vee \neg z$ and $\neg x \vee \neg y \vee z$ (as well as both constants). A simple PCSP template decidable by $\blp$ is e.g. $\pcsp \ab$ where $A$, $B$ are ordered sets, $|A| \leq |B|$, and $\bfa$ and $\bfb$ both have a single binary relation defined by $x < y$. A well-known class of templates with $\blp$-decidable VCSPs are those that contain only submodular valued relations (see~\cite{Kolmogorov13,KZ_vcsps}). Finally, the 2-approximation of the Vertex Cover problem~\cite{KZ_vcsps} is a PVCSP decidable by $\blp$. In all the mentioned examples, it is not hard to find symmetric (fractional) polymorphisms of every arity.
\end{example}

A \textit{fractional homomorphism} \cite{ThapperZ12,ViolaZ21} from $\bfa$ to $\bfb$  is a unary fractional polymorphism of $\ab$, or equivalently, a probability distribution $\mu$ over $B^{A}$ such that for every $R \in \sig$ and every $\textbf{a} \in A^{\ar(R)}$ we have that
\begin{displaymath} 
\sum_{f \in B^{A}} \mu(f) R^\bfb(f(\textbf{a})) \leq R^\bfa(\textbf{a}).
\end{displaymath}

If there exists a fractional homomorphism from $\bfa$ to $\bfb$, we say that $\bfa$ is  fractionally homomorphic to $\bfb$ and we write $\bfa \to_f \bfb$. 

The following is a characterization of PVCSP templates in terms of fractional homomorphisms.

\begin{proposition} \label{prop:FH}
For any two valued $\sig$-structures $\bfa$ and $\bfb$, the following are equivalent.
\begin{enumerate}
    \item There exists a fractional homomorphism from $\bfa$ to $\bfb$. \label{item:1FH}
    \item For all non-negative finite-valued $\sig$-structures $\bfx$, $\opt(\bfx, \bfb) \leq \opt(\bfx, \bfa)$.\label{item:2FH}
    \end{enumerate}
\end{proposition}
\begin{proof} (\ref{item:1FH}) $\Rightarrow$ (\ref{item:2FH}). Let $\mu$ be a fractional homomorphism from $\bfa$ to $\bfb$, let $g:X \to A$ be such that $\opt(\bfx,\bfa) = \val(\bfx,\bfa,g)$, and let $f \in B^A$ be some map that minimizes $\val(\bfx,\bfb,f\circ g)$. Then \begin{align*}
    \opt(\bfx,\bfb) & \leq \val(\bfx,\bfb,f\circ g)  \leq \sum_{f' \in B^A} \mu(f') \val(\bfx,\bfb,f'\circ g) \\
    & = \sum_{R \in \sig} \sum_{\bx \in X^{\ar(R)}} R^\bfx(\bx) \sum_{f' \in B^A} \mu(f') R^\bfb(f' \circ g(\bx))\\
   & \leq \sum_{R \in \sig} \sum_{\bx \in X^{\ar(R)}} R^\bfx(\bx) R^\bfa(g(\bx)) = \val(\bfx,\bfa,g) = \opt(\bfx,\bfa).
\end{align*}

$(\ref{item:2FH}) \Rightarrow (\ref{item:1FH})$. The idea for this proof is to assume that there is no fractional homomorphism from $\bfa$ to $\bfb$, formulate this fact as infeasibility of a system of linear inequalities, and then use a version of Farkas' Lemma~\cite{farkas} to find a valued structure $\bfx$ with $\opt\xb > \opt \xa$.

The existence of a fractional homomorphism from $\bfa$ to $\bfb$  can be reformulated as the following system of linear inequalities, where there is a rational-valued variable $\mu_f$ for every $f \in B^A$.
\begin{subequations}\label{eq:fh1-program}
\begin{align}
\textnormal{variables: }& \quad \mu_f  \textnormal{ for all } f \in B^A  \nonumber \\
\textnormal{constraints: }& \quad \sum_{f \in B^A} \mu_f R^{\bfb}(f(\textbf{a})) \leq R^{\bfa}(\textbf{a})  \quad \textnormal{ for all } R \in \sig \textnormal{ and } \textbf{a} \in A^{\ar(R)} \label{eq:fh1} \\
& \quad \sum_{f \in B^A} \mu_f \geq 1 \\
& \quad \mu_f \geq 0  \quad \textnormal{for all } f \in B^A.
\end{align}
\end{subequations}

If there is no fractional homomorphism from $\bfa$ to $\bfb$, the system (\ref{eq:fh1-program}) is infeasible. 

We now deal with infinite coefficients.
Define $B^A_{<\infty}= \{f \in B^A: \forall R \in \sig, \ \forall \ba \in A^{\ar(R)}, \  \RA(\ba) < \infty \textnormal{ implies } \RB(f(\ba))<\infty\}$.
Now consider the new linear system obtained from (\ref{eq:fh1-program}) by first removing all the inequalities in (\ref{eq:fh1}) where $\RA(\ba)=\infty$ (since these inequalities are always satisfied), and second, by removing from (\ref{eq:fh1-program}) the variable $\mu_f$ for all $f \in B^A \setminus B^A_{<\infty}$ and changing (\ref{eq:fh1}) so that the sums run over  $B^A_{<\infty}$ only (since we need to have $\mu_f=0$ for $f \in B^A \setminus B^A_{<\infty}$ in any feasible solution). Clearly, the system of linear inequalities resulting from this procedure remains infeasible and does not contain infinite coefficients.

This system of linear inequalities can be rewritten in matrix form as $M \textbf{f} \leq \textbf{a}$ subject to $\textbf{f} \geq 0$, where $\textbf{f} \in \mathbb{Q}_{\geq}^{ B^A_{<\infty}}$ is the vector of unknowns, and $M$ is a real-valued matrix. By Farkas' Lemma, if the program (\ref{eq:fh1-program}) is not feasible, then
the system of inequalities $M^T \textbf{y} \geq 0$ subject to $\textbf{a}^T \textbf{y} < 0$ and $\textbf{y} \geq 0$ is feasible. Explicitly, the latter system is the following.
\begin{subequations}
\begin{align} \label{eq:BTy0}
 \textnormal{variables: }& \quad y, x_{R,\ba} \textnormal{ for every } R \in \sig \textnormal{ and }\textbf{a} \in A^{\ar(R)} \textnormal{ with } \RA(\ba)< \infty \nonumber \\
 \textnormal{constraints: }& \quad  \sum_{R \in \sig}\sum_{\substack{\textbf{a} \in A^{\ar(R)}\\ \RA(\ba)< \infty}} x_{R,\ba} R^{\bfb}(f(\textbf{a})) \geq y \quad \quad \textnormal{for all } f \in  B^A_{<\infty} \\ 
    & \quad \sum_{R \in \sig} \sum_{\substack{\textbf{a} \in A^{\ar(R)}\\ \RA(\ba)< \infty}} x_{R,\ba} R^{\bfa}(\textbf{a}) < y \\
    & \quad y \geq 0, \ x_{R,\ba} \geq 0 \quad \textnormal{for all } R \in \sig , \textbf{a} \in A^{\ar(R)}.
\end{align}
\end{subequations}
Eliminating $y$, and adding trivially satisfied constraints to (\ref{eq:BTy0}) for all $f \in B^A \setminus B^A_{<\infty}$, we get that the following system is feasible.

\begin{subequations} \label{eq:BTy0reduced-program}
\begin{align}
    \textnormal{variables: }& \quad \mathrlap{ x_{R,\ba} \textnormal{ for every } R \in \sig \textnormal{ and }\textbf{a} \in A^{\ar(R)} \textnormal{ with } \RA(\ba)< \infty }\nonumber \\
    \textnormal{constraints: }& \quad \sum_{R \in \sig} \sum_{\substack{\textbf{a} \in A^{\ar(R)}\\ \RA(\ba)< \infty}} x_{R,\ba} R^{\bfb}(f(\textbf{a})) >  \sum_{R \in \sig} \sum_{\substack{\textbf{a} \in A^{\ar(R)}\\ \RA(\ba)< \infty}} x_{R,\ba} R^{\bfa}(\textbf{a}) \textnormal{ for all } f \in  B^A   \label{eq:BTy0reduced-ineq1}\\ 
    & \quad \mathrlap{x_{R,\ba} \geq 0 \quad \textnormal{for all } R \in \sig, \ \textbf{a} \in A^{\ar(R)}.}
\end{align}
\end{subequations}

Let $x_{R,\ba}$ for $R \in \sig$, $\ba \in A^r$ be a feasible solution to (\ref{eq:BTy0reduced-program}), and consider the structure $\bfx$ with domain $X=A$ and relations given by $R^{\bfx}(\textbf{a})=x_{R,\ba}$ for $\textbf{a} \in A^{\ar(R)}$ with $\RA(\ba)< \infty$ and $R^{\bfx}(\textbf{a})=0$ whenever $\RA(\ba)= \infty$. Notice that $\bfx$ is non-negative finite-valued, that the right-hand side in (\ref{eq:BTy0reduced-ineq1}) is equal to $\val(\bfx,\bfa,\mathrm{id})$ (where $\mathrm{id}$ denotes the identity function), and that the left-hand side is equal to $\val(\bfx,\bfb,f)$. 
Therefore $\opt\xb > \opt \xa$, as required.
\end{proof}

The valued version of the decomposition theorem for $\sa^1$ uses a concept that is ``dual'' to  fractional homomorphism, as suggested by Proposition~\ref{prop:dualFH} below. Even though a similar result has appeared elsewhere (see \cite[Proposition 3.6]{Carbonnel2022otherside}), we provide a proof for the sake of completeness. In fact, only the forward implication $(\ref{item:1dualFH}) \Rightarrow (\ref{item:2dualFH})$ is needed for our results.

We define a \textit{dual fractional homomorphism} from $\bfx$ to $\bfy$ ($\bfx \to_{df} \bfy$) to be a probability distribution $\eta$ over $Y^X$ such that for every $R \in \sig$ and every $\by \in Y^{\ar(R)}$ we have that
\begin{displaymath}
 R^{\bfy}(\by) \geq \sum_{f \in Y^X} \eta(f) \sum_{\substack{\bx \in X^{\ar(R)} \\ \by = f(\bx)}} R^{\bfx}(\bx).\end{displaymath}
\begin{proposition}[\cite{Carbonnel2022otherside}] \label{prop:dualFH} For any two non-negative finite-valued $\sig$-structures $\bfx$ and $\bfy$, the following are equivalent.
\begin{enumerate}
    \item There exists a dual fractional homomorphism from $\bfx$ to $\bfy$. \label{item:1dualFH}
    \item For all valued $\sig$-structures $\bfa$, $\opt(\bfx,\bfa) \leq \opt(\bfy,\bfa)$. \label{item:2dualFH}
    \end{enumerate}
\end{proposition}

\begin{proof}
$(\ref{item:1dualFH}) \Rightarrow (\ref{item:2dualFH})$. Let $\eta$ be a dual fractional homomorphism from $\bfx$ to $\bfy$, and $g:Y \to A$ be such that $\opt(\bfy,\bfa) = \val(\bfy,\bfa,g)$. Then
\begin{align*}
    \opt(\bfy,\bfa) & =  \sum_{R \in \sig} \sum_{\by \in Y^{\ar(R)}} R^\bfy(\by) R^\bfa(g(\by)) \\
     & \geq \sum_{R \in \sig} \sum_{\by \in Y^{\ar(R)}} \sum_{f \in Y^X} \eta(f) \sum_{\substack{\bx \in X^{\ar(R)} \\ \by = f(\bx)}} R^{\bfx}(\bx)  R^\bfa(g \circ f(\bx)) \\
     & = \sum_{f \in Y^X} \eta(f) \val(\bfx,\bfa,g \circ f),
\end{align*}
which implies that there exists some function $f':X \to Y$ such that $ \val(\bfx,\bfa,g \circ f') \leq  \opt(\bfy,\bfa) $, hence $ \opt(\bfx,\bfa) \leq  \opt(\bfy,\bfa)$ as required. Notice that this holds regardless of whether $\bfa$ is finite-valued or general-valued. 

$(\ref{item:2dualFH}) \Rightarrow (\ref{item:1dualFH})$. 
The contrapositive is proved in a similar way to Proposition~\ref{prop:FH}. The proof is in fact somewhat simpler since we do not need to deal with infinite values. 

The existence of a dual fractional homomorphism from $\bfx$ to $\bfy$ can be reformulated as the following linear program: 

\begin{subequations} \label{eq:dualfh1-program}
\begin{align}
  \textnormal{variables: }& \quad \eta_f  \textnormal{ for all }  f \in Y^X \nonumber\\
  \textnormal{constraints: }& \quad  \sum_{f \in Y^X} \eta_f \sum_{\substack{\bx \in X^r\\ \by=f(\bx)}} R^\bfx(\textbf{x}) \leq R^{\bfy}(\by) \quad \textnormal{for all } R \in \sig \textnormal{ and } \textbf{y} \in Y^{\ar(R)}  \label{eq:dualfh1} \\
& \quad \sum_{f \in Y^X} \eta_f \geq 1 \\
& \quad \eta_f \geq 0 \quad \textnormal{for all } f \in Y^X.
\end{align}
\end{subequations}

Assume that there is no dual fractional homomorphism from $\bfx$ to $\bfy$. Then, the linear system (\ref{eq:dualfh1-program}) is infeasible.

Notice that $\bfx$ and $\bfy$ are assumed to be finite-valued, so all the coefficients in this system of linear inequalities are finite-valued too. Then, (\ref{eq:dualfh1-program}) can be rewritten in matrix form as $M \textbf{f} \leq \by$ subject to $\textbf{f} \geq 0$, where $\textbf{f} \in \mathbb{Q}_{\geq}^{Y^X}$ is the vector of unknowns, and $M$ is a real-valued matrix.  By Farkas' Lemma, if (\ref{eq:dualfh1-program}) is not feasible, then the linear program $M^T \textbf{z} \geq 0$ subject to $\textbf{y}^T \textbf{z} < 0$ and $\textbf{z} \geq 0$ is feasible. Explicitly, this program is the following.
\begin{subequations}
\begin{align} \label{eq:dualBTy0}
\textnormal{variables: }& \quad t, z_{R,\bu} \textnormal{ for every } R \in \sig \textnormal{ and }\textbf{y} \in Y^{\ar(R)}  \nonumber \\
\textnormal{constraints: }& \quad  \sum_{R \in \sig}\sum_{\textbf{y} \in Y^{\ar(R)}} \mkern-8mu z_{R,\by} \mkern-12mu \sum_{\substack{\bx \in X^{\ar(R)}\\ \by=f(\bx)}} \mkern-8mu R^{\bfx}(\textbf{x}) \geq t \quad  \textnormal{for all } f \in  Y^X  \\
    & \quad \sum_{R \in \sig} \sum_{\textbf{y} \in Y^{\ar(R)}} z_{R,\by} R^{\bfy}(\textbf{y}) < t\\
    & \quad  t \geq 0, \ z_{R,\by} \geq 0 \quad \textnormal{for all } R \in \sig , \textbf{y} \in Y^{\ar(R)}.
\end{align}
\end{subequations}
Eliminating $t$ and rearranging we get that the following program is feasible.
\begin{subequations} \label{eq:dualBTy0reduced-program}
\begin{align} \label{eq:dualBTy0reduced}
    \textnormal{variables: }& \mathrlap{\quad z_{R,\by} \textnormal{ for every } R \in \sig \textnormal{ and }\textbf{y} \in Y^{\ar(R)}} \nonumber \\
    \textnormal{constraints: }& \quad  \sum_{R \in \sig}\sum_{\textbf{y} \in Y^{\ar(R)}} \mkern-8mu z_{R,\by} \mkern-8mu \sum_{\substack{\bx \in X^{\ar(R)}\\ \by=f(\bx)}}  \mkern-8mu R^{\bfx}(\textbf{x}) >   \sum_{R \in \sig} \sum_{\textbf{y} \in Y^{\ar(R)}} \mkern-12mu z_{R,\by} R^{\bfy}(\textbf{y}) \quad \textnormal{ for all } f \in  Y^X \\
    & \mathrlap{\quad z_{R,\by} \geq 0 \quad \textnormal{ for } R \in \sig , \textbf{y} \in Y^{\ar(R)}.}
\end{align}
\end{subequations}

Let $z_{R,\by}$ for $R \in \sig$, $\by \in Y^{\ar(R)}$ be a feasible solution to (\ref{eq:dualBTy0reduced-program}), and consider the valued structure $\bfa$ with domain $A=Y$ and valued relations given by $R^\bfa(\textbf{y})=z_{R,\by}$ for each $\textbf{y} \in Y^{\ar(R)}$. 
For all $f \in Y^X$, we have
\begin{align*}
\val(\bfy,\bfa,\mathrm{id}) & = \sum_{R \in \sig} \sum_{\textbf{y} \in Y^{\ar(R)}} R^\bfa(\textbf{y}) R^{\bfy}(\textbf{y}) \\ &
< \sum_{R \in \sig}\sum_{\textbf{y} \in Y^{\ar(R)}} R^\bfa(\textbf{y}) \sum_{\substack{\bx \in X^{\ar(R)}\\ \by=f(\bx)}} R^{\bfx}(\textbf{x})= \val(\bfx,\bfa,f),
\end{align*}
therefore $\opt(\bfy,\bfa) < \opt(\bfx,\bfa)$, which is the negation of (\ref{item:2dualFH}).
\end{proof}

\subsection{The decomposition theorem for valued structures}
Finally, we state the decomposition theorem for valued structures. This will provide a connection between the combinatorial and the LP-based characterizations of the class of Weisfeiler-Leman invariant PVCSPs.

\begin{restatable}{theorem}{thmDecomp} \label{thm:Decomp}
Let $\bfx$, $\bfa$ be a pair of similar valued structures, where $\bfx$ is non-negative and finite-valued. 
Then there exist non-negative finite-valued structures $\yone,\ytwo$ such that
\begin{enumerate}
    \item $\bfx \to_{df} \yone$,
    \item $\yone \eqone \ytwo$, and
    \item $\opt(\ytwo,\bfa)\leq \opt^{\sa^1}\xa$.
\end{enumerate}
\end{restatable}

\begin{proof}
If $\saxa$ is not feasible, then we can take $\yone=\ytwo=\bfx$, and the statement follows trivially, so from now on we shall assume that $\saxa$ is feasible. In this case, the construction of the valued structures $\yone$ and $\ytwo$ is analogous to the construction of Theorem \ref{thm:SA1}, except that one has to additionally take into consideration the weights of the constraints in $\bfx$. In particular, let $m>0$ be the least common denominator of a rational solution to $\sa^1\xa$. For every constraint $R(\bx) \in \con\bfx$ of arity $r$, let 
$\rho_1, \ldots, \rho_r$ be the permutations of $[m]$ obtained as in the proof of Theorem \ref{thm:SA1} so that they satisfy conditions (\ref{property:1}) and (\ref{property:2}) in said proof. Then, for each $k \in [m]$ we set $R^{\yone}((k,\bx[1]),\ldots,(k,\bx[r])) = R^{\ytwo}((\rho_1(k),\bx[1]),\ldots,(\rho_r(k),\bx[r]))=1/m \cdot R^{\bfx}(\bx)$, and all the other tuples in $\yone$ and $\ytwo$ are given zero weight.

Then $\bfx \to_{df} \yone$ by the dual fractional homomorphism given by the uniform distribution over $f_k$, $k \in [m]$, where $f_k: X \to Y_1$ is defined by $f_k(x) = (k,x)$ for all $x \in X$.

On the other hand, let $h: Y_2 \to A$ be the same map that gives a homomorphism from $\bfy_2$ to $\bfa$ in the crisp setting.
For each $R(\bx) \in \con\bfx$ the weights of the tuples that correspond to $R(\bx)$ are selected so that their contribution to $\val(\ytwo,\bfa,h)$ is equal to the inner sum in the objective function (\ref{eq:objBLP}) of $\sa^1\xa$; therefore, the total value of $h$ for $(\ytwo,\bfa)$ is equal to $\opt^{\sa^1}\xa$. It follows that $\opt(\ytwo,\bfa)\leq \opt^{\sa^1}\xa$.

Moreover, the iterated degree of $(k,x)$ in both $\yone$ and $\ytwo$ is obtained from the iterated degree of $x$ by scaling down each label $(S,R,q)$ to $(S,R,q/m)$. It follows that $\yone \eqone \ytwo$, and the proof is concluded.\end{proof}

As in the crisp setting, the dual fractional homomorphism $\bfx \to_{df} \yone$, the equivalence $\yone \eqone \ytwo$, and the assignment $Y_2 \to A$ that witness $\opt(\ytwo,\bfa)$ from the proof of Theorem~\ref{thm:Decomp} can all be associated with rational matrices of the appropriate dimensions, the product of which is associated to a solution to $\saxa$. Hence, again, Theorem~\ref{thm:Decomp} can be regarded as a decomposition theorem.

\subsection{Characterizing Weisfeiler-Leman invariant PVCSPs} 

Similarly to the crisp case, Theorem \ref{thm:Decomp} can be applied to characterize the promise valued templates $\ab$ for which $\pvcsp\ab$ is Weisfeiler-Leman invariant.

\begin{theorem} \label{th:main-pvcsp}
Let $\ab$ be a promise valued template of signature $\sig$. Then the following are equivalent.
\begin{enumerate}
    \item $\pvcsp\ab$ is closed under $\eqone$-equivalence; \label{item:2thmain-pvcsp}
    \item $\sa^1$ decides $\pvcsp\ab$. \label{item:3thmain-pvcsp}
\end{enumerate}
\end{theorem}

\begin{proof}

 The forward direction $(\ref{item:2thmain-pvcsp}) \Rightarrow (\ref{item:3thmain-pvcsp})$ is an immediate consequence of the decomposition theorem for valued structures.
 We need to show that for every non-negative finite-valued $\sig$-structure $\bfx$, $\opt \xb \leq \opt^{\sa^1}\xa$.
 Let $\yone, \ytwo$ be the structures obtained from Theorem~\ref{thm:Decomp}, i.e., $\bfx \to_{df} \yone$, $\yone \eqone \ytwo$, and $\opt(\ytwo,\bfa) \leq \opt^{\sa^1}\xa$. Then, by (\ref{item:2thmain-pvcsp}) we have that $\opt(\yone,\bfb) \leq \opt^{\sa^1}\xa$, and by Proposition \ref{prop:dualFH}, $\opt(\bfx,\bfb) \leq \opt^{\sa^1}\xa$ too as required (see Figure \ref{fig:decomp-proof} for a diagram of this proof).

\begin{figure}
\begin{minipage}[c]{0.6\textwidth}
\vspace{-1.25cm}
\hspace{-.75cm}
   \begin{tikzcd}
     & \yone \arrow[rr, "\eqone" description, no head, dotted] \arrow[rrrr,"\leq \opt^{\sa^1}\!\!\xa\!" description, bend left ] && \ytwo \arrow[rddd, "\leq \opt^{\sa^1}\!\!\xa\!" description] &      & \bfb &\\
     &&&&&& \\
     &&&&&& \\
\bfx \arrow[rrrr, "\saxa" description, rightsquigarrow] \arrow[ruuu, "df" description, dashed] \arrow[rrrrruuu,"\leq \opt^{\sa^1}\!\!\xa\!" description, bend left=100, looseness=1.2]& & & & \bfa \arrow[ruuu, "f" description, dashed] & &\\
\end{tikzcd}
  \end{minipage}\hfill
  \begin{minipage}[c]{0.35\textwidth}
    \caption[Diagram of the proof in Theorem \ref{th:main-pvcsp}.]{Diagram of the proof of (\ref{item:2thmain-pvcsp}) $\Rightarrow$ (\ref{item:3thmain-pvcsp}) of Theorem \ref{th:main-pvcsp}. Fractional and dual fractional homomorphisms are pictured as dashed arrows, the equivalence relation $\eqone$ as a dotted line, the feasibility of the linear program $\sa^1$ as a squiggly arrow, and an upper bound on the optimum value for a pair of structures as a label on a standard arrow.} \label{fig:decomp-proof}
  \end{minipage}  
    \end{figure}
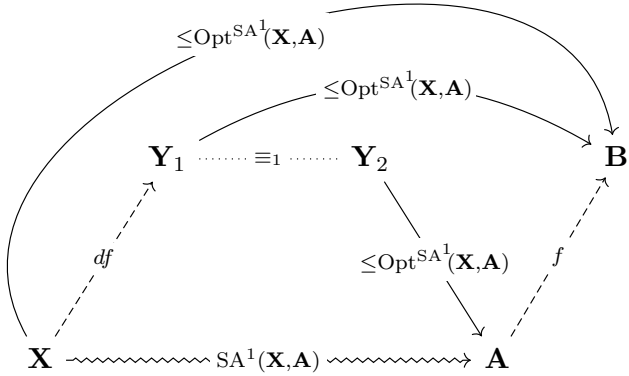

For $(\ref{item:3thmain-pvcsp}) \Rightarrow (\ref{item:2thmain-pvcsp})$, we cannot use the same simple argument as in the crisp case, due to the fact that in the valued setting, feasible solutions to linear programs cannot be composed to obtain meaningful upper bounds on the objective functions. Fortunately, we can work around this. Consider an input valued structure $\bfx$ with $\opt^{\sa^1}\xa < \infty$ and let $p$ be a feasible solution of $\saxa$. Notice that the solution $\overline{p}$ obtained by taking the average of $p$ over the equivalence classes of $\eqone$ is still feasible and in particular it achieves the same optimal value as $p$. It follows that if $\saxa$ is feasible, then there is always an optimal feasible solution that assigns the same value to every class of $\eqone$-equivalent variables and constraints of $\bfx$. If we apply this reasoning to $\sa^1(\bfx_1\cup \bfx_2,\bfa)$, where $\bfx_1\cup \bfx_2$ denotes the disjoint union of two $\eqone$-equivalent structures $\bfx_1$ and $\bfx_2$, we can deduce that $\sa^1(\bfx_1,\bfa)$ and $\sa^1(\bfx_2,\bfa)$ have the same optimal value. Since $\sa^1$ decides $\pvcsp\ab$, it follows in particular that that $\opt(\bfx_2,\bfb) \leq \opt(\bfx_1,\bfa)$ as required.
\end{proof}

Clearly, the implication (\ref{4-main-csp}) $\Rightarrow$ (\ref{3-main-csp}) in Theorem~\ref{th:main-csp} remains true for the PVCSP (so, by virtue of Theorem \ref{th:main-pvcsp}, the PVCSPs in Example~\ref{ex:pvcsp} are all Weisfeiler-Leman invariant). The following example shows that, unlike for CSPs, the converse implication does not hold in general: we provide a $\pvcsp$ template that is decided by $\sa^1$ but not by $\blp$.

\begin{example} \label{ex:SAnotBLP}
Let $\bfa$, $\bfb$ be $\sig$-structures where $\sig$ contains a single binary relation symbol $R$. Let $A=B=\{0,1\}$,  $R^\bfa(a,a)=R^\bfb(a,a)=3$ for $a \in \{0,1\}$, and $R^\bfa(a,b)=2$, $R^\bfb(a,b)=0$ for $a \neq b \in \{0,1\}$. The probability distribution which assigns probability 1 to the identity function is a fractional homomorphism, so $\ab$ is a $\pvcsp$ template.

We claim that $\blp$ does not decide $\pvcsp\ab$. Indeed, let $\bfx$ be the $\pvcsp$ input structure given by $X=\{x\}$ and $R^{\bfx}(x,x)=1$. Then, there is a feasible solution to $\blp\xa$ given by $p_x(a)=1/2$ for $a \in \{0,1\}$ and $p_{R(x,x)}(a,a)=0$, $p_{R(x,x)}(a,b)=1/2$ for $a \neq b \in \{0,1\}$. This solution witnesses that $\opt^\blp\xa \leq 2$, however, it is easy to see that $\opt\xb =3$ and so $\blp$ does not decide $\pvcsp\ab$.

On the other hand, we show that  $\opt\xb \leq \opt^{\sa^1}\xa$ for any input valued structure $\bfx$. Let $V_l(\bfx) = \sum_{x \in X} R^\bfx(x,x)$ and $V_e(\bfx) = \sum_{x \neq y} R^\bfx(x,y)$ be the total weight of the constraints in $\bfx$ with and without repetitions, respectively. We choose an assignment $h:X \to B$ at random: each $h(x)$ is chosen independently and uniformly (both 0 and 1 with probability $1/2$).
The expected value of $\val(\bfx,\bfb,h)$ is $3 V_l(\bfx) + 3/2V_e(\bfx) $, which implies that $\opt\xb\leq 3 V_l(\bfx) + 3/2V_e(\bfx)$. 
As for $\sa^1$, we know that any feasible solution must have $p_{(x,x)}(a,b)=0$ whenever $a \neq b$. Therefore, we get
\begin{align*}
    \opt^{\sa^1}\xa & = \min \Big[ \sum_{x \in X}  \sum_{a \in A}  p_{R(x,x)}(a,a) R^\bfx(x,x) R^\bfa(a,a)   + \\ & \  \sum_{x \neq y \in X} \sum_{a,b \in A} p_{R(x,y)}(a,b) R^\bfx(x,y) R^\bfa(a,b) \Big]
 \geq  3 V_l(\bfx)  + 2 V_e(\bfx) > \opt\xb.
\end{align*}
\end{example}

\section{Hierarchies of relaxations} \label{chap:higher}

We start by discussing the $k$-dimensional Weisfeiler-Leman method for graphs and its connection to logic and bounded treewidth. We then define a notion of higher-dimensional equivalence for relational structures and show that it is deeply connected to both the Sherali-Adams and Weisfeiler-Leman hierarchies.

\subsection{$k$-WL, Counting Logic, and Treewidth} \label{sec:kWL}

To address the limitations of the basic colour refinement algorithm, an increasingly powerful hierarchy of relaxations based on the Weisfeiler-Leman method, known as the \textit{$k$-dimensional Weisfeiler-Leman algorithm} ($k$-WL), has been developed independently by multiple researchers.

As in the 1-dimensional case, $k$-WL proceeds in iterations. At every iteration $j$ on a graph $G=(V;E)$, $k$-WL produces a colouring $\chi^{k}_j$ of the set of $k$-tuples of $V$.
In particular, for $\bv \in V^k$, $\chi^{k}_0(\bv)$ is given by the isomorphism type of $\bv$: that is, two $k$-tuples $\bu=(u_1,\ldots,u_k)$ and $\bv=(v_1,\ldots,v_k)$ have the same isomorphism type if and only if the mapping
$u_i \mapsto v_i$ is an isomorphism between the subgraphs of $G$ induced by $\{u_1,\ldots,u_k\}$ and $\{v_1,\ldots,v_k\}$ respectively. For $j \geq 1$, $k$-WL sets $$\chi^{k}_{j+1}(\bv) = (\chi^{k}_j(\bv), \mathcal{M}_{j}(\bv))$$ where for $k \geq 2$ and $i\in[k]$, 
\begin{align*}
    \mathcal{M}_j(\bv) &= \{\{\big(\chi^{k}_j(\zeta_1(\bv,w)), \ldots,\chi^{k}_j(\zeta_k(\bv,w))\big) \mid w \in V\}\},\\
     \zeta_i(\bv,w)&=(v_1,\ldots,v_{i-1},w,v_{i+1},\ldots,v_k)
\end{align*}
and for $k=1$, $$\mathcal{M}_j(\bv) = \{\{\big(\chi^{k}_j(w) \mid \{v,w\} \in E\}\},$$ so that 1-WL is precisely the colour refinement algorithm.

As before, a stable colouring is eventually reached, and two graphs are said to be distinguished by $k$-WL if their stable colourings differ. Notice that for all practical purposes (for example in terms of isomorphism testing), it is enough to compute the partition induced by the stable colouring instead of the actual colours viewed as multisets, which can quickly become very large. In this setting, the stable $k$-WL colouring of an $n$-vertex graph can be computed in time $\mathcal{O}(n^{k+1} \log n)$ \cite{Immerman1990}.

It had initially been conjectured that
the $k$-dimensional Weisfeiler-Leman algorithm would provide a polynomial time isomorphism test for graphs of bounded degree. In their seminal paper,  Cai, F\"{u}rer, and Immerman disproved this conjecture by constructing for every integer $k$ a pair of non-isomorphic graphs on $\mathcal{O}(k)$ vertices which cannot be distinguished by $k$-WL (yet they are distinguished by $(k+1)$-WL).

\begin{theorem}[\cite{Cai92}]\label{th:CFI} 
For all $k \in \mathbb{N}$, there exist two graphs on $\mathcal{O}(k)$ vertices such that they are not distinguished by $k\textnormal{-WL}$.
\end{theorem}

The proof of indistinguishability of the graphs from Theorem \ref{th:CFI}, which came to be known as the \textit{CFI construction}, relies on the equivalence between $k$-WL and the counting logic $\counting^{k+1}$, which is the other main contribution in \cite{Cai92}.

On the other hand, it is known from an old result of Lov{\'a}sz \cite{lovasz1967operations} that counting homomorphisms characterizes isomorphism classes: that is, two graphs $G$ and $H$ are isomorphic if and only if for every graph $F$, $F$ has the same number of homomorphisms into $G$ and into $H$. Theorem \ref{th:fractional-graph} shows how restricting this condition to counting homomorphisms from trees yields a combinatorial characterization of indistinguishability by 1-WL. It turns out that considering not just trees but structures of bounded treewidth corresponds to indistinguishability under $k$-WL \cite{Dvorak2010,Dell2018}.

Let $\mathscr{P}(A)$ denote the power set of $A$. A {\em tree-decomposition}~\cite{Robertson1984} of a structure $\bfa$ is a pair $(G,\beta)$ where $G=(V,E)$ is a tree and $\beta:V\rightarrow\mathscr{P}(A)$ is a mapping such that the following conditions are satisfied:
\begin{enumerate}
    \item For every constraint $R(\ba)$ in $\con\bfa$ there exists a node $v\in V$ such that $\{\ba\}\subseteq\beta(v)$;
    \item If $a\in\beta(u)\cap \beta(v)$ then $a\in\beta(w)$ for every node $w$ in the unique path in $G$ joining $u$ to $v$.
    \end{enumerate}
The {\em width} of a tree-decomposition $(G,\beta)$ is $\max\{|\beta(v)| \mid v\in V\}-1$ and the \textit{treewidth} of $\bfa$ is defined as the smallest width among all its tree-decompositions. It is easy to see that the only graphs that have treewidth 1 are trees and forests. Moreover, a $\sig$-ftree has treewidth at most $r-1$, where $r$ is the maximum arity of a relation in $\sig$.

Then we have that the logical and combinatorial characterizations of colour refinement from Theorem \ref{th:fractional-graph} can be lifted to higher dimensions, as summarized in the following result:

\begin{theorem}[\cite{Cai92,Dvorak2010,Dell2018}] \label{th:kWL} Let $G$, $H$ be graphs. The following are equivalent:
\begin{enumerate}
    \item $k\textnormal{-WL}$ does not distinguish $G$ from $H$.
    \item $G$ and $H$ satisfy the same formulae in the logic $\counting^{k+1}$;
    \item $\hom(T;G)=\hom(T;H)$ for all graphs $T$ of treewidth at most $k$.
\end{enumerate}
\end{theorem}

Recall that indistinguishability under 1-WL is equivalent to the feasibility of the fractional isomorphism linear program (see item (\ref{item:1fr}) in Theorem \ref{th:fractional-graph}). Surprisingly, Atserias and Maneva \cite{Atserias2013} and Malkin \cite{Malkin2014} were able to lift this correspondence to higher dimensions: in particular, they showed that a hierarchy of linear programs obtained by applying the Sherali-Adams method to the system of linear equations defining fractional isomorphism interleaves with indistinguishability under $k$-WL. Subsequently, in \cite{Grohe2015Pebble} Grohe and Otto gave a variant of the SA-based relaxation of isomorphism whose levels correspond tightly with $k$-WL, and showed that the interleaving of the two hierarchies from \cite{Atserias2013,Malkin2014} is strict.

\subsection{Sherali-Adams meets Weisfeiler-Leman}

Motivated by this correspondence between the Sherali-Adams hierarchy applied to fractional isomorphism and the Weisfeiler-Leman algorithm, we ask the question: can a similar correspondence 
be established when applying the Sherali-Adams method to relaxations of homomorphism (instead of isomorphism) of arbitrary relational structures (instead of graphs)? Our results seem to indicate that the answer to this question is positive. In a nutshell, our approach is based on the observation (formalized in 
Lemma \ref{le:starkSherali}) that the feasibility of $\sa^k\xa$ is equivalent to the feasibility of $\sa^1(\Phi(\bfx),\Phi(\bfa))$ where $\Phi(\cdot)$ is a suitable transformation that maps every relational structure to a binary structure. More specifically, both the domain and the relations of the output structure $\Phi(\bfa)$,
which we shall denote $\astarkbf$, are defined from $\bfa$ via a primitive-positive formula. In a similar spirit, \cite{Ciardo2022Tensors} analyzes $\sa^k$ via a ``reduction'' to $\sa^1$ obtained via tensorization. 

More specifically, for every $k>0$ we define an operator $\stark$ that maps a $\sig$-structure to a new structure whose signature we will denote $\sig^*_k$. 
Let $\bfa$ be a $\sig$-structure. Then we define the universe of $\astarkbf$ to be $\astark := \cup_{j \leq k} A^j \cup \con\bfa.$
Additionally, for $\ba \in A^j$ and ${\bf i}=(i_1,\dots,i_n)\in [j]^n$, we use $\pr_{\bf i} \ba$ to denote the projection of $\ba$ to $\bf i$, i.e, the tuple 
$(a_{i_1},\dots,a_{i_n})$. Then, $\astarkbf$ is defined to contain the following unary ($T_{j,S}$, $R_S$) and binary ($\tjii$, $R_{\bi}$) relations:
\begin{spreadlines}{8pt}
\begin{align*}
T_{j,S}^{\astarkbf}\,\,&=\{ \ba\in A^j \mid a_i=a_{i'} \: \forall i,i'\in S\}  &&  j\leq k, \, S\subseteq[j]\\
\tjii^{\astarkbf}\,\,&=\{(\textbf{a},\pii\textbf{a}) \mid \textbf{a} \in A^j\}  &&  j',j\leq k, \, \textbf{i}\in [j]^{j'}\\
R_S^{\astarkbf}\,\,&=\{ R(\ba) \mid \ba \in R^\bfa, a_i=a_{i'} \: \forall i,i'\in S\} &&  R\in\sig, \, S\subseteq[\ar(R)]\\
R_{\bi}^{\astarkbf}\,\,&=\{ (R(\ba),\pii\ba) \mid \ba\in R^\bfa\} &&  R\in\sig, \, j\leq k, \, \textbf{i}\in [\ar(R)]^{j}.
\end{align*}
\end{spreadlines}
We remark that $\astarkbf$ can be seen as a multi-sorted pp-power of $\bfa$ with domains $A^1,\ldots, A^k$, $R^{\bfa}_1, R^{\bfa}_2,\ldots$ (notice $R^{\bfa}_i \subseteq A^{\ar(R_i)}$).

To begin with, thanks to the operator $\stark$, we are able to reduce feasibility of $\sa^k$ to feasibility of $\sa^1$.

\begin{restatable}{lemma}{lestarkSherali} \label{le:starkSherali}
Let $\bfx$, $\bfa$ be $\sig$-structures. Then $\sa^k\xa$ is feasible if and only if $\sa^1(\xstarkbf,\astarkbf)$ is feasible.
\end{restatable}

\begin{proof}
    The proof is purely syntactical, although it is convenient to first slightly modify the LP formulations $\sak$ and $\sastar$. We shall refer to the solutions of $\sak$ and $\sastar$ by appropriately indexed sets of variables $p$, $q$ respectively. 

\begin{itemize}
\item In $\sak$, it follows from (\ref{eq:SA4}) that we can safely replace all variables $p_{R(\bx)}(f)$ with $f(\textbf{x}) \not\in R^\bfa$ by $0$.
\item In $\sastar$ some more substantial observations are needed. First, for each $j\leq k$ and each $\bx\in X^j$, it follows from conditions (\ref{eq:SA3}) and (\ref{eq:SA4}) for
$T_{j,S}$ ($S\subseteq[j])$ that for every $v$ in $\astark$, $q_{\bx}(v)$ must take value $0$ unless $v=f(\bx)$ for some function $f:\{\bx\}\rightarrow A$. Hence, in a first stage we set $q_{\bx}(v)$ to zero for each $j\leq k$, each $\bx\in X^j$
and each $v$ that is not a tuple of the form $f(\bx)$ for some function $f:\{\bx\}\rightarrow A$.

Furthermore, it follows from condition (\ref{eq:SA3}) for $\tjii$ that $q_{\bx}(f(\bx))=q_{\bx'}(f(\bx'))$ for every $\bx$, $\bx'$ satisfying $\{\bx\}=\{\bx'\}$ and every $f:\{\bx\}\rightarrow A$. Hence, in a second stage,
for each $V\subseteq X$ with $|V|\leq k$ and every $f:V\rightarrow A$ we identify all variables $q_{\bx}(f(\bx))$ which satisfy 
$\{\bx\}=V$. 

Then, consider now the variables of the form $q_{R(\bx)}(v)$, $v\in\astark$. It
follows from conditions (\ref{eq:SA3}) and (\ref{eq:SA4}) for $R_S$ ($S\subseteq[\ar(R)])$ that
$q_{R(\bx)}(v)$ must be set to $0$ unless $v=R(f(\bx))$ for some function $f:\{\bx\}\rightarrow A$.

The other variables in $\sastar$ are of the form $q_C(f)$ where $C\in \con{\xstarkbf}$. As we shall see they can always safely be identified with some of the other variables. Let us start first with the case in which $C$ is a unary constraint. If $C=T_{j,S}(\bx)$ or $C=R_S(\bx)$, then it follows from (\ref{eq:SA3}) that $q_C(f)=q_{\bx}(f(\bx))$. Assume
now that $C$ is a binary constraint, that is $C=\tjii(\bx,\pii \bx)$ or $C=R_{\bf i}(\bx,\pii \bx)$. It follows again from 
(\ref{eq:SA3}) that $q_C(f)=q_{\bx}(f(\bx))$.
\end{itemize}

Now we are ready to prove the lemma. In particular, consider the following one-to-one correspondence between
the assignments in $\sak$ and $\sastar$:
\begin{itemize}
\item Every variable $p_V(f)$ of $\sak$ is assigned to the variable $q_{\bx}(f(\bx))$ of $\sastar$, where $\bx$ is any tuple satisfying $\{\bx\}=V$.
\item Every variable $p_{R(\bx)}(f)$ of $\sak$ is assigned to the variable $q_{R(\bx)}(R(f(\bx)))$ of $\sastar$.
\end{itemize}
It is not difficult to see that this correspondence preserves feasibility.
\end{proof}

As an immediate consequence of Lemma \ref{le:starkSherali} together with the Theorem \ref{thm:SA1}, we obtain a decomposition of the full hierarchy of Sherali-Adams relaxations for the homomorphism problem in terms of fractional isomorphism.

\begin{theorem} \label{thm:SAk}
Let $\bfx$, $\bfa$ be relational structures. Then, the following are equivalent:
\begin{enumerate}
    \item $\sa^k(\bfx,\bfa)$ is feasible;
    \item There exists a sequence of structures $\bfy_0, \ldots, \bfy_n$ such that $\bfy_0=\xstarkbf$, $\bfy_n=\astarkbf$, and for all $i=0,\ldots,n-1$ we have that $\bfy_i \to \bfy_{i+1}$ or $\bfy_{i} \eqone \bfy_{i+1}$;
    \item There exists a pair of structures $\yone, \ytwo$ such that $\xstarkbf \to \bfy_{1}$, $\bfy_{1} \eqone \bfy_{2}$, and $\ytwo \to \astarkbf$. 
\end{enumerate}
\end{theorem}

Furthermore, the operator $\stark$ allows us to define, similarly to the case of graphs, a hierarchy of increasingly tighter relaxations of isomorphism for relational structures. In particular, for every $k> 1$, we shall denote $\bfa \eqwl[k] \bfb$ whenever $\astarkbf \eqone \bstarkbf$. 

Then, Theorem \ref{thm:SAk} implies that feasibility of $\sa^k$ (for a fixed right-hand structure $\bfa$) is closed under $\eqwl[k]$.
\begin{corollary} \label{cor:sa-k-eqwl-k}
Let $\bfx$, $\bfx'$ and $\bfa$ be $\sig$-structures and $k \geq 1$. Suppose that $\bfx\eqwl[k]\bfx'$. Then, $\sa^k\xa$ is feasible  if and only if $\sa^k(\bfx',\bfa)$ is feasible.
\end{corollary}

For $k=1$, it is easy to see that $\bfa^*_1$ and $\bfb^*_1$ satisfy the conditions of Theorem $\ref{le:digraph}$ precisely when $\bfa$ and $\bfb$ do, and hence, $\eqwl[k]$ correctly extends $\eqone$. For other small values of $k$ other than $1$, the characterization of $\equiv_k$ might not be so straightforward. However, as long as $k$ is at least as large as the arity of any relation in the signature, then we have the following characterization that parallels an analogous result for $k$-WL  (i.e., Theorem \ref{th:kWL}). 

\begin{restatable}{theorem}{homTreewidth} \label{th:homTreewidth}
Let $r$ be the maximum arity among all relations in $\sig$ and assume that $r\leq k$. Then for every pair of structures $\bfa$, $\bfb$ the following are equivalent:
\begin{enumerate}
    \item $\bfa \eqwl[k] \bfb$; \label{1homTreewidthItem}
    \item $\hom(\bfx;\bfa)=\hom(\bfx;\bfb)$ for every $\sig$-structure $\bfx$ of treewidth $<k$;\label{2homTreewidthItem}
    \item $\bfa$ and $\bfb$ satisfy the same formulae in the logic $\counting^k$. \label{3homTreewidthItem-logic}
\end{enumerate}

\end{restatable}

We note that $\equiv_k$ is {\em not} intended to be a strict generalization of $k$-WL to relational structures as $k$-WL and $\equiv_k$ do not necessarily coincide for graphs. Still, a comparison of Theorems \ref{th:homTreewidth} and Theorem \ref{th:kWL} reveals that $\equiv_k$ and $k$-WL are intimately related. In particular, in the case of graphs, $\equiv_k$ corresponds to $(k-1)$-WL for every $k\geq 2$.

We conclude this chapter by providing  proofs of Theorem \ref{le:digraph} and Theorem \ref{th:homTreewidth}.

\subsection{Proof of Theorem \ref{le:digraph}}

\ledigraph*

Before starting we note that two structures $\bfa$ and $\bfb$ have a common equitable partition if
there is an equitable partition $(\partp,\partq)=(\{\ipartp_i \mid i\in I\},\{\ipartq_j \mid j\in J\})$ of the disjoint union
$\bfa\cup \bfb$ such that 
\begin{enumerate}
    \item $|\ipartp_i\cap A|=|\ipartp_i\cap B|$ for every $i\in I$, and \label{1conditionPart}
    \item $|\ipartq_j\cap \con\bfa|=|\ipartq_j\cap \con\bfb|$ for every $j\in J$. \label{2conditionPart}
\end{enumerate}
Note that since $(\partp,\partq)$ is equitable, then any of conditions (\ref{1conditionPart}) or (\ref{2conditionPart}) implies the other. We shall freely switch between the two alternative definitions of common equitable partition. 

We also need a few additional definitions before embarking on the proof. 
We denote by $J_V$  the $V \times V$ matrix whose entries are all ones and we use $\oplus$ to denote direct sums of matrices. Let $M$ be a $V\times W$ matrix and consider two subsets $S_{V} \subseteq V$, $S_{W} \subseteq W$. The restriction of $M$ to $S_{V}$ and $S_{W}$, denoted by $M[S_{V},S_{W}]$, is the matrix obtained by removing from $M$ all the rows that do not belong in $S_{V}$ and all the columns that do not belong in $S_{W}$.

A matrix $M \in \mathbb{R}^{V\times W}$ is \textit{decomposable} if there exists a partition $(V_1,V_2)$ of $V$ 
and a partition $(W_1,W_2)$ of $W$ such that for every two distinct $i,j\in [2]$ and every $v\in V_i$ and $w\in W_j$, $M_{v,w}=0$. In this case we can write $M=M_1\oplus M_2$ where $M_i=M[V_i,W_i]$. Otherwise, $M$ is said to be \textit{indecomposable}.

Let $\matleft\in[0,1]^{V\times V}$ be a doubly stochastic matrix. Note that $\matleft$ has a unique decomposition $\matleft=\oplus_{i\in I} \matleft_i$ where each $\matleft_i$ is an indecomposable doubly stochastic matrix. The \textit{row partition} of $\matleft$ is defined to be the partition of $V$ into classes $\ipartp_i,i\in I$ where $\ipartp_i$ contains the rows of $\matleft_i$, and the \textit{column partition} is defined in an analogous manner.

We will also need the following lemma (see for example \cite{Scheinerman2011fractional}).

\begin{lemma} \label{le:indecomposable}
Let $\matleft\in \mathbb{R}^{V \times V}$, $\matright\in \mathbb{R}^{W \times W}$ be doubly stochastic indecomposable matrices.
\begin{enumerate}
    \item\label{1indecomposableItem} Let  ${\bf a},{\bf b}\in\mathbb{R}^{V}$ such that $\matleft\cdot {\bf a}={\bf b}$ and $\matleft^T\cdot {\bf b}={\bf a}$. Then, there exists $c\in\mathbb{R}$ such that \[{\bf a}={\bf b}=c\cdot {\bf 1.}\] 
    \item\label{2indecomposableItem} Let $M_1,M_2\in\mathbb{R}^{V\times W}$ such that $\matleft M_1=M_2\matright$ and $M_1\matright^T=\matleft^T M_2$. \label{miitem}
Then, there exist $c,d\in\mathbb{R}$ such that the following identities hold:
\[M_1{\bf 1}=M_2{\bf 1}=c\cdot{\bf 1} \hspace{2cm} {\bf 1}^T M_1={\bf 1}^T M_2=d \cdot {\bf 1}^T\]
\end{enumerate}
\end{lemma}

\begin{proof}
Item (\ref{1indecomposableItem}) is Theorem 6.2.4 (ii) from \cite{Scheinerman2011fractional}. We enclose a proof of item (\ref{2indecomposableItem}). Let ${\bf a}=M_1\cdot {\bf 1}$ and ${\bf b}=M_2\cdot {\bf 1}$. We have $\matleft M_1\cdot {\bf 1}=M_2 \matright\cdot{\bf 1}=M_2\cdot {\bf 1}$ since $\matright$ is doubly stochastic, which implies $\matleft\cdot {\bf a}={\bf b}$. Similarly, we have $M_1\cdot {\bf 1}=M_1 \matright^T\cdot {\bf 1}=\matleft^T M_2\cdot {\bf 1}$ which implies ${\bf a}=\matleft^T\cdot {\bf b}$. Then, we apply item (\ref{1indecomposableItem}) to deduce that there exists $c \in \mathbb{R}$ such that $M_1{\bf 1}=M_2{\bf 1}=c\cdot{\bf 1}$. The other identity is proven in an analogous way.
\end{proof}

 $(\ref{1digraphItem})\Rightarrow (\ref{3digraphItem})$. Let $\matleft$ and $\matright$ be doubly stochastic matrices satisfying (\ref{1digraphItem}). Let $\matleft=\oplus_{i\in I} \matleft_i$, $\matright=\oplus_{j\in J} \matright_j$ be decompositions of $\matleft$ and $\matright$. Denote by $(\partp^{\bfa},\partp^{\bfb})$ the column and row partitions of $\matleft$ respectively and by $(\partq^{\bfa},\partq^{\bfb})$ the column and row partitions of $\matright$.
Then, the restrictions of $M_{\bfa}^{\ell}$ and $M_{\bfb}^{\ell}$ to $(\ipartp^{\bfa}_{i},\ipartq^{\bfa}_{j})$, $(\ipartp^{\bfb}_{i},\ipartq^{\bfb}_{j})$ respectively  satisfy 
\[\matleft_i M_{\bfa}^{\ell}[\ipartp^{\bfa}_{i},\ipartq^{\bfa}_{j}]= M_{\bfb}^{\ell}[\ipartp^{\bfb}_{i},\ipartq^{\bfb}_{j}] \matright_j,\]
\[M_{\bfa}^{\ell}[\ipartp^{\bfa}_{i},\ipartq^{\bfa}_{j}]\matright_j^T= \matleft_i^T M_{\bfb}^{\ell}[\ipartp^{\bfb}_{i},\ipartq^{\bfb}_{j}].\]

Hence it follows from Lemma \ref{le:indecomposable} that there exist $c^{\ell}_{i,j}$, $d^{\ell}_{i,j}$ such 
that 
\begin{align*}
&M_{\bfa}^{\ell}[\ipartp^{\bfa}_{i},\ipartq^{\bfa}_{j}]\cdot{\bf 1}=M_{\bfb}^{\ell}[\ipartp^{\bfb}_{i},\ipartq^{\bfb}_{j}]\cdot{\bf 1}=c^{\ell}_{i,j}\cdot {\bf 1} \\
&{\bf 1}^T \cdot M_{\bfa}^{\ell}[\ipartp^{\bfa}_{i},\ipartq^{\bfa}_{j}]={\bf 1}^T \cdot M_{\bfb}^{\ell}[\ipartp^{\bfb}_{i},\ipartq^{\bfb}_{j}]={\bf 1}^T\cdot d^{\ell}_{i,j} 
\end{align*}
which is equivalent to conditions (\ref{eq:eqpartQ}) and (\ref{eq:eqpartP}) showing that partitions $(\partp^{\bfa},\partq^{\bfa})$ and $(\partp^{\bfb},\partq^{\bfb})$ have the same parameters.

$(\ref{3digraphItem})\Rightarrow (\ref{2digraphItem})$. Assume that $(\{\ipartp^{\bfa}_i \mid i\in I\},\{\ipartq^{\bfa}_j \mid j\in J\})$ and $(\{\ipartp^{\bfb}_i \mid i\in I\},\{\ipartq^{\bfb}_j \mid j\in J\})$ define a common equitable partition of $\bfa$ and $\bfb$. We shall prove that any two elements that are in the same set of the partition of $\bfa$ and $\bfb$ must have the same iterated degree. That is, we show by induction on $k$ that for all $k\geq 0$, $\delta_k^{\bfa}(a)=\delta_{k}^{\bfb}(b)$ whenever there exists $i\in I$ such that $a\in \ipartp^{\bfa}_i$ and $b\in \ipartp^{\bfb}_i$ (the case $R(\textbf{a}) \in \ipartq^{\bfa}_j,R(\textbf{b}) \in \ipartq^{\bfb}_j$ is analogous).
The base case ($k=0$) is immediate. For the inductive case, assume that the statement holds for $k-1$. 
Let $(\ell,\delta)$ be any arbitrary element in $\delta^{\bfa}_k(a)$. We shall show that it has the same multiplicity in $\delta^{\bfa}_k(a)$ and in $\delta^{\bfb}_{k}(b)$. By the inductive hypothesis it follows that there exists $J_{\delta}\subseteq J$ such that 
\[\{R(\textbf{a}) \in \con\bfa \mid \delta^{\bfa}_{k-1}(R(\textbf{a}))=\delta\}=\bigcup_{j\in J_{\delta}} \ipartq_{j}^{\bfa}\]
\[\{R(\textbf{b})\in  \con\bfb \mid \delta^{\bfb}_{k-1}(R(\textbf{a}))=\delta\}=\bigcup_{j\in J_{\delta}} \ipartq_{j}^{\bfb}\]
Then, the multiplicity of $(\ell,\delta)$ in $\delta^{\bfa}_k(a)$ is
\begin{equation}\label{eq:multiplicity}
|\{R(\textbf{a}) \in\bigcup_{j\in J_{\delta}} \ipartq_{j}^{\bfa} \mid  M_{\bfa}[a,R(\textbf{a})]=\ell\}|.
\end{equation}
Since $(\partp^{\bfa},\partq^{\bfa})$ is an equitable partition it follows that (\ref{eq:multiplicity}) is equal to $\sum_{j\in J_{\delta}} c_{i,j}^{\bfa,\ell}$
where $c_{i,j}^{\bfa,\ell}$ are the parameters of the partition.

It can analogously be shown that the the multiplicity of $(\ell,\delta)$ in $\delta^{\bfb}_k(b)$ is
$\sum_{j\in J_{\delta}} c_{i,j}^{\bfb,\ell}$.  Since $(\partp^{\bfa},\partq^{\bfa})$ and $(\partp^{\bfb},\partq^{\bfb})$ define
a common equitable partition it follows that $c^{\bfa,\ell}_{i,j}=c^{\bfb,\ell}_{i,j}$ for every $i\in I,j\in J$ and we are done.

$(\ref{2digraphItem})\Rightarrow (\ref{3digraphItem})$. Assume that $\bfa$ and $\bfb$ have the same iterated degree sequence. Let $\partp=\{\ipartp_i \mid i\in I\}$ and $\partq=\{\ipartq_j \mid j\in J\}$ be the partition of $A\cup B$ and $\con\bfa\cup \con\bfb$ induced by the fixed point $\delta^{\bfa \cup \bfb}$. It is easy to verify that $(\partp,\partq)$ defines a common equitable partition.

$(\ref{3digraphItem})\Rightarrow (\ref{1digraphItem})$. Assume that 
$(\{\ipartp^{\bfa}_i \mid i\in I\},\{\ipartq^{\bfa}_j \mid j\in J\})$ and $(\{\ipartp^{\bfb}_i \mid i\in I\},\{\ipartq^{\bfb}_j \mid j\in J\})$ define a common equitable partition. For ease of notation it is convenient that matrices $M^{\ell}_\bfa$ and $M^{\ell}_\bfb$
are indexed by the same sets of rows $V$ and columns $W$, which can be done by fixing a one-to-one correspondence
between $A$, $B$, and $V$ and similarly between $\con\bfa$, $\con\bfb$, and $W$. Furthermore, this correspondence can be established
in such a way that under 
it $\partp^{\bfa}$ and $\partp^{\bfb}$ become the same partition $\partp$ over $V$, and $\partq^{\bfa}$ and $\partq^{\bfb}$ become the same partition $\partq$ over $W$.

Define $\matleft = \oplus_{i \in I} \frac{1}{|\ipartp_{i}|} J_{\ipartp_i}$ and $\matright = \oplus_{j \in J} \frac{1}{|\ipartq_{j}|} J_{\ipartq_j}$. Clearly $\matleft$ and $\matright$ are doubly stochastic. It remains to show that $\matleft M_{\bfa}^{\ell}=M_{\bfb}^{\ell}\matright$ and $M_{\bfa}^{\ell}\matright^{T}=\matleft^{T}M_{\bfb}^{\ell}$ for all labels $\ell \in \LAB_{\sig}$. 
Now since $(\partp,\partq)$ is an equitable partition, it follows that for every $i \in I$, $j \in J$, and $\ell \in \LAB_{\sig}$ there exist parameters $c^{\ell}_{ij}$ and $d^{\ell}_{ji}$ which satisfy conditions (\ref{eq:eqpartQ}) and (\ref{eq:eqpartP}).

Then for all $i \in I$, $j \in J$ and $\ell \in \LAB_{\sig}$ it holds that 
\begin{equation*}
|\ipartp_{i}| c^{\ell}_{ij} = |\ipartq_{j}|d^{\ell}_{ji},
\end{equation*}
and that the sum of the elements of the respective $\ipartp_{i}\times \ipartq_{j}$ portions of $M_{\bfa}^{\ell}$ and $M_{\bfb}^{\ell}$ are equal. Now let $b \in \ipartp_{i} \cap B$ and $R(\textbf{a}) \in \ipartq_{j} \cap \con\bfa$. Then 
\begin{align*}
    (\matleft M_{\bfa}^{\ell})[b,R(\textbf{a})] = \frac{1}{|\ipartp_{i}|} d_{ji}^{\ell} = \frac{1}{|\ipartq_{j}|} c_{ij}^{\ell} = (M_{\bfb}^{\ell}\matright)[b,R(\textbf{a})], 
\end{align*}
showing that $\matleft M_{\bfa}^{\ell} = M_{\bfb}^{\ell}\matright$ for all labels $\ell \in \LAB_{\sig}$ as required, and similarly, noting that $\matleft^{T}=\matleft$ and $\matright^{T}=\matright$, we obtain that $M_{\bfa}^{\ell}\matright^{T} = \matleft^{T}M_{\bfb}^{\ell}$ for all $\ell \in \LAB_{\sig}$ too.

$(\ref{1digraphItem})\Rightarrow (\ref{5digraphItem})$. Assume that $\matleft$ and $\matright$ satisfy (\ref{1digraphItem}). We shall prove that $\matleft$ satisfies (\ref{5digraphItem}).

Let $R$ be the unique binary edge relation symbol in $\sig$. We define two matrices $U_{\bfa}$, $Z_{\bfa}$ where $U_{\bfa}=(M_{\bfa}^{(1,R)}+M_{\bfa}^{(2,R)})/\sqrt{2}$ and $Z_{\bfa}=2M_{\bfa}^{(1,R)}/\sqrt{2}$.
Then we have that $\na=U_{\bfa} U_{\bfa}^T-Z_{\bfa} Z_{\bfa}^T$. We similarly obtain $\nb=U_{\bfb} U_{\bfb}^T- Z_{\bfb}Z_{\bfb}^T$ if we define $U_{\bfb}$ and $Z_{\bfb}$ accordingly. 
Note that the identities of (\ref{1digraphItem}) are still satisfied if we replace $M_{\bfa}^{\ell}$ and $M_{\bfb}^{\ell}$ by $U_{\bfa}$
and $U_{\bfb}$ or $Z_{\bfa}$
and $Z_{\bfb}$ respectively.

It then follows that
\[\matleft U_{\bfa} U_{\bfa}^T= U_{\bfb} \matright U_{\bfa}^T = U_{\bfb} U_{\bfb}^T \matleft. \] and \[\matleft Z_{\bfa} Z_{\bfa}^T= Z_{\bfb} \matright Z_{\bfa}^T = Z_{\bfb} Z_{\bfb}^T \matleft. \] 
Therefore, $\matleft \na = \nb \matleft$ as required.

$(\ref{5digraphItem})\Rightarrow (\ref{3digraphItem})$. Let $R$ be the unique binary edge relation symbol in $\sig$. Let $\matleft$ be a doubly stochastic matrix satisfying $\matleft\na= \nb \matleft$. Since $\na$ and $\nb$ are symmetric it also holds that $\na \matleft^T= \matleft^T \nb$.  Let $\matleft=\oplus_{i\in I} \matleft_i$ be a decomposition of $\matleft$ and denote by $\{\ipartp_i^{\bfa} \mid i\in I\}$ and $\{\ipartp_i^{\bfb} \mid i\in I\}$ the column and row partitions of $\matleft$. Applying the same reasoning as in $(\ref{1digraphItem})\Rightarrow (\ref{3digraphItem})$ it follows that for every $i,i'\in I$, there exists $c_{i,i'}$ such that for every $\bfd\in \{\bfa,\bfb\}$ and every $d\in \ipartp_i^{\bfd}$
\[\{ d'\in \ipartp^{\bfd}_{i'} \mid R(d,d')\in \con\bfa\}=c_{i,i'}  \]
It then easily follows that if we let $\partq^{\bfa}$, $\partq^{\bfb}$ be the partitions of $\con\bfa$, $\con\bfb$ respectively given by assigning two edges to the same partition class if the  partition classes of their vertices coincide then $(\partp^{\bfa},\partq^{\bfa})$ and $(\partp^{\bfb},\partq^{\bfb})$ define
a common equitable partition of $\bfa$ and $\bfb$.

\subsection{Proof of Theorem \ref{th:homTreewidth}}

In order to prove Theorem \ref{th:homTreewidth}, it will be useful to introduce a simple combinatorial game, which is a variant of the Ehrenfeucht-Fra\"{i}ssé games \cite{Fraisse1955,Ehrenfeucht1961}, that characterizes equivalence in $\counting^k$.

The \textit{bijective $k$-pebble game} is played by two players, $\spo$ and $\dup$, by placing $k$ pairs of pebbles on a pair of structures $\bfa$, $\bfb$ of the same size (if $\bfa$ and $\bfb$ have different sizes, $\spo$ is always assumed to win the game). We shall denote each pair of pebbles by $(x_i,y_i)$, where $x_i$ belongs to $\spo$ and $y_i$ belongs to $\dup$. At every round, $\spo$ picks up a pebble $x_i$, and $\dup$ picks up the corresponding pebble $y_i$. At this point, $\dup$ chooses a bijection $f$ between $A$ and $B$. Then, $\spo$ places a pebble $x_i$ on an element $a \in A$, and $\dup$ must place the corresponding pebble $y_i$ on $f(a) \in B$.

$\dup$ wins a round of the bijective $k$-pebble game if the partial map defined by $x_i \mapsto y_i$ (i.e., where the element of $A$ under pebble $x_i$ is mapped to the element of $B$ under pebble $y_i$) is a partial isomorphism between (the respective vertex-induced substructures of) $\bfa$ and $\bfb$. Otherwise, $\spo$ wins the round.
We say that $\dup$ has a \textit{winning strategy} for the bijective $k$-pebble game if she has a strategy to win every round of the game (note that the game has infinitely many rounds).

It is a well-known result of Hella \cite{Hella1996hierarchies} (see also \cite{Grohe2015Pebble,Abramsky2017pebbling}) that two structures $\bfa$ and $\bfb$ satisfy the same formulae in the logic $\counting^k$ if and only if Duplicator has a winning strategy for the bijective $k$-pebble game on $\bfa$, $\bfb$.

On the other hand, the extension of the Lov{\'a}sz-type result of Dvo\v{r}{\'{a}}k from graphs to relational structures (i.e., the equivalence of (\ref{2homTreewidthItem}) and (\ref{3homTreewidthItem-logic}) in Theorem \ref{th:homTreewidth}) was shown by Dawar, Jakl and Reggio \cite{DawarJR21} using a categorical approach based on the game comonads of Abramsky, Dawar and Wang \cite{Abramsky2017pebbling}. Additionally, an independent proof of this result was given by Butti and Dalmau \cite{ButtiD21fractional} using purely combinatorial techniques.

Therefore, all together, we have the following result.

\begin{theorem} \label{th:tree+logic}
Let $r$ be the maximum arity among all relations in $\sig$ and assume that $r\leq k$. Then for every pair of structures $\bfa$, $\bfb$ the following are equivalent:
\begin{enumerate}
    \item $\bfa \eqwl[k] \bfb$; \label{item:eqwl}
    \item $\hom(\bfx;\bfa)=\hom(\bfx;\bfb)$ for every $\sig$-structure $\bfx$ of treewidth $<k$;\label{1treewidth}
    \item $\bfa$ and $\bfb$ satisfy the same formulae in the logic $\counting^k$. \label{item:logic}
    \item Duplicator has a winning strategy for the bijective $k$-pebble game on $\bfa$, $\bfb$. \label{item:pebble}
\end{enumerate}
\end{theorem}

\begin{proof}
By virtue of the above discussion, it only remains to prove the equivalence of (\ref{item:eqwl}) and (\ref{item:pebble}).
First, let us make the definition of winning strategy a little more formal.

For a $j$-tuple $\bd =(d_1,\ldots,d_j)$, an element $d$, and $i \in [j+1]$, define $\bd^i_d = (d_1,\ldots,d_{i-1},d,d_{i},\ldots,d_j)$. Additionally, for $j \in [k]$ and $i \in [j]$, define $\bi(j,i)=(1,\ldots,i-1,i+1,\ldots,j)$. 
Then, a winning strategy for the bijective $k$-pebble game is a non-empty set $W \subseteq \cup_{0 \leq j \leq k} (A^j \times B^j)$ such that, for every $\ba=(a_1,\ldots,a_j)$ and $\bb=(b_1,\ldots,b_j)$ with $(\ba,\bb) \in W$, the following conditions hold:
\begin{enumerate}[(i)]
    \item\label{strat:1} The partial map given by $a_i \mapsto b_i$, $i \in [j]$ is a partial isomorphism between $\bfa$ and $\bfb$;
    \item\label{strat:2} If $j \geq 1$, then $(\pi_{\bi(j,i)}\ba,\pi_{\bi(j,i)}\bb) \in W$ for every $i \in [j]$;
  \item\label{strat:3} If $j<k$, there exists a bijection $f:A \to B$ such that $(\ba^i_a,\bb^i_{f(a)})\in W$ for every $a \in A$ and every $i \in [j+1]$. 
\end{enumerate}
In what follows, we will assume without loss of generality that $A$ and $B$ are disjoint, so when we talk about the union of $\bfa$ and $\bfb$, we will always mean the \textit{disjoint} union.

$(\ref{item:eqwl}) \Rightarrow (\ref{item:pebble})$.
It will be convenient to assume that $\astark$ and $\bstark$ additionally include a tuple of null arity (this will correspond to the configuration where no pebbles are placed yet). Let $(\partp,\partq)$ be an equitable partition of $\astarkbf \cup \bstarkbf$. A winning strategy $W$ for $\dup$ can be obtained as the set of all pairs $(\ba,\bb) \in \astark \times \bstark$ such that $\ba$ and $\bb$ belong to the same class of $\partp$.

To see that $W$ is a winning strategy, first observe that, by the definition of the transformation $\stark$, any two tuples of $\astarkbf \cup \bstarkbf$ in the same partition class of $\partp$ must have the same isomorphism type in the original structure $\bfa \cup \bfb$. It follows that for all elements $(\ba,\bb) \in W$, the map $a_i \mapsto b_i$ is a partial isomorphism from $\bfa$ to $\bfb$, hence $W$ satisfies condition $(\ref{strat:1})$. 

Second, it is easy to see that $W$ satisfies condition (\ref{strat:2}), since for each $j \in [k]$ and $\bd \in A^j \cup B^j$, $\bd$ has a unique out-neighbor of type $T_{j,\bi(j,i)}$, which is precisely $\pi_{\bi(j,i)}\bd$. Therefore, for every $(\ba,\bb) \in W$, $\pi_{\bi(j,i)}\ba$ and $\pi_{\bi(j,i)}\bb$ must belong to the same class of $\partp$ and hence $(\pi_{\bi(j,i)}\ba,\pi_{\bi(j,i)}\bb) \in W$. 

Finally, to see that $W$ satisfies condition (\ref{strat:3}), for each $(\ba,\bb) \in W$ of arity $j < k$ we need to find a bijection $f:A \to B$ in such a way that for every $a \in A$ and every $i \in [j+1]$, the tuples $\ba^i_a$ and $\bb^i_{f(a)}$ belong to the same partition class of $\partp$.
Since $\ba$ and $\bb$ belong to the same class of $\partp$ and  $(\partp,\partq)$ is equitable, there is a one-to-one correspondence that respects the partition between the constraints that $\ba$ and $\bb$ participate in; in particular, this holds for the constraints of type $T_{j+1,\bi(j+1,i)}$. Since for each $a \in A$ and $b \in B$, $T_{j+1,\bi(j+1,i)}(\ba^i_{a},\ba) \in \con{\astarkbf}$ and $T_{j+1,\bi(j+1,i)}(\bb^i_{b},\bb) \in \con{\bstarkbf}$, there is a one-to-one partition-preserving correspondence between tuples $\ba^i_{a}\in A^{j+1}$, $\bb^i_{b} \in B^{j+1}$ that induces an appropriate bijection between $A$ and $B$. In particular, the unary constraints of $\astarkbf$, $\bstarkbf$ of the form $T_{j+1,S}$, $S \subseteq[j+1]$ imply that $f(a_i)=b_i$ for each $i \in [j]$.

$(\ref{item:pebble}) \Rightarrow (\ref{item:eqwl})$. We begin with the simple observation that, since $\ar(R) \leq k$ for all $R \in \sig$, the transformation $\stark$ can be substantially simplified without its basic properties being affected (particularly, Lemma \ref{le:starkSherali} and the definition of $\eqwl[k]$ remain unchanged). Specifically, this simplification consists of the following changes:
\begin{enumerate}
    \item \label{obs:1pebbles} For each $R \in \sig$, identify $R(\ba)$ with $\ba$, so the universe of $\astarkbf$ is just $\cup_{j \leq k}A^j$;
    \item \label{obs:2pebbles} Remove from $\sig^\stark$ all binary relation symbols of the form $R_\bi$ (since for all $\textbf{i}\in [\ar(R)]^{j}$ and all $\ba \in R^{\bfa}$, $R_\bi(\ba,\pii\ba)$ is equivalent to $T_{\ar(R),\bi}(\ba,\pii\ba) \land R_\emptyset(\ba)$); 
    \item \label{obs:3pebbles} Remove from $\sig^\stark$ all unary relation symbols of the form $R_S$ for all $S \neq \emptyset$ (since for all $S \subseteq [\ar(R)]$ and all $\ba \in R^\bfa$, $R_S(\ba)=T_{\ar(R),S}(\ba) \land R_\emptyset(\ba)$);
    \item \label{obs:4pebbles} Remove from $\sig^\stark$ all binary relation symbols of the form $\tjii$ whenever $|[j] \setminus \{\bi\}|> 1$. This is because, if $|[j] \setminus \{\bi\}|> 1$, then there exists a sequence $\bi_1 \in [j_0]^{j_1}$, $\bi_2 \in [j_1]^{j_2}$, $\ldots \ $, $\bi_m \in [j_{m-1}]^{j_m}$ with $j_0=j$ and $m \geq 2$ such that $|[j_{i-1}] \setminus \{\bi_i\}| \leq 1$ for each $i \in [m]$, and $\pi_{\bi_m} \ldots \pi_{\bi_2} \pi_{\bi_1}= \pii$. Then, we have that \begin{align*}
        \tjii(\ba,\pii\ba) & = T_{j,\bi_{1}}(\ba,\pi_{\bi_1}\ba) \land T_{j_1,\bi_{2}}(\pi_{\bi_1}\ba,\pi_{\bi_2}\pi_{\bi_1}\ba) \land \ldots \\
        & \ldots \land T_{j_{m-1},\bi_{m}}(\pi_{\bi_{m-1}}\ldots\pi_{\bi_1}\ba,\pi_{\bi_{m}} \pi_{\bi_{m-1}}\ldots \pi_{\bi_1}\ba).
    \end{align*}
\end{enumerate}

Now, assume that $\dup$ has a winning strategy for the bijective $k$-pebble game. We need to find an equitable partition of $\astarkbf \cup \bstarkbf$ such that each class of the partition has the same number of elements from $\astark$ and $\bstark$.
Since $\astarkbf \cup \bstarkbf$ is a (vertex- and edge-coloured) digraph, it is enough to define a partition $\partp=\{\ipartp_i \mid i \in I\}$ of $\astark \cup \bstark$ where any two elements in the same partition class of $\partp$ are subject to the same unary constraints, and have the same number of in- and out-neighbors in any other class of $\partp$ connected by an edge of any given colour.

Let $W$ be a winning strategy for $\dup$. We define the set $\overline{W}$ to contain precisely all the pairs $(\ba,\bb) \in \cup_{j \leq k} A^j \times B^j$ such that there exist $\ba_0,\ba_1,\ldots,\ba_n \in A^j$ and  $\bb_0,\bb_1,\ldots,\bb_n \in B^j$ such that $\ba_0=\ba$, $\bb_n=\bb$,  $(\ba_0,\bb_0) \in W$, and for each $i \in [n]$, the pairs $(\ba_i,\bb_{i-1})$ and $(\ba_i,\bb_i)$ both belong to $W$. That is, we think of $\overline{W}$ as the closure of $W$ under odd chains (see Figure \ref{fig:W-odd-chains}).

\begin{figure}
     \centering
\begin{tikzpicture}
    \node at (0,0) (n1) {$\ba_0$};
    \node at (0,2) (n2) {$\bb_0$};
    \node at (2,0) (n3) {$\ba_1$};
    \node at (2,2) (n4) {$\bb_1$};
    \node at (4,0) (x) {$\ldots$};
    \node at (4,2) (y) {$\ldots$};
    \node at (6,0) (n7) {$\ba_n$};
    \node at (6,2) (n8) {$\bb_n$};
    \draw[-latex, line width=0.25mm] (n1)--(n2);
    \draw[-latex, line width=0.25mm] (n3)--(n2);
    \draw[-latex, line width=0.25mm] (n3)--(n4);
    \draw[-latex, line width=0.25mm] (2.5,1.5)--(n4);
    \draw[dashed, line width=0.25mm] (3,1)--(2.5,1.5);
    \draw[line width=0.25mm] (n7)--(5.5,0.5);
    \draw[dashed, line width=0.25mm] (5.5,0.5)--(5,1);
    \draw[-latex, line width=0.25mm] (n7)--(n8);
    \draw[-latex, dashed, gray, line width=0.15mm] (n1)--(n4);
    \draw[-latex, dashed, gray, line width=0.15mm] (n1)--(n8);
    \draw[-latex, dashed, gray, line width=0.15mm] (n3)--(n8);
    \draw[-latex, dashed, gray, line width=0.15mm] (n7)--(n2);
    \draw[-latex, dashed, gray, line width=0.15mm] (n7)--(n4);
\end{tikzpicture}
     \caption[A winning strategy closed under odd chains.]{A representation of the closure of $\overline{W}$ under a chain of length $2n+1$ in $W$. Pairs in $W$ are represented as solid black arrows and the added pairs in $\overline{W}$ are represented as dashed gray arrows.}
     \label{fig:W-odd-chains}
 \end{figure}
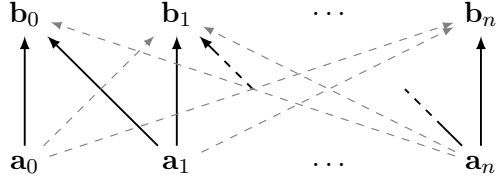

We shall show that $\overline{W}$ is also a winning strategy for the bijective $k$-pebble game over $\bfa$, $\bfb$. Conditions (\ref{strat:1}) and (\ref{strat:2}) are easily verified by using the corresponding conditions for $W$, plus the fact that isomorphisms compose. For (\ref{strat:3}), let $(\ba,\bb) \in \overline{W}$ have arity $j$ and let $\ba_0,\bb_0,\ldots,\ba_n,\bb_n$ be the odd chain witnessing this. Let $f_{i,i}$, $f_{i+1,i}$ be the bijections given by condition (\ref{strat:3}) of $W$ for $(\ba_i,\bb_i)$, $(\ba_{i+1},\bb_i)$ respectively. We claim that $f:=f_{n,n} \circ f^{-1}_{n,n-1} \circ \ldots \circ f_{1,1} \circ f_{1,0}^{-1} \circ f_{0,0}$ witnesses that (\ref{strat:3}) is satisfied for $(\ba_0,\bb_n)=(\ba,\bb)$ in $\overline{W}$. 
Recall that for a $j$-tuple $\bd$, an element $d$, and $l \in [j+1]$, we denote by $\bd^l_d$ the tuple obtained from $\bd$ by adding $d$ to $\bd$ before the $l\ith$ coordinate. 
By condition (\ref{strat:3}) for $W$, we have that for every $l \in [j]$, every $i=0,\ldots,n$, and every $\alpha_i \in A$, 
$$((\ba_i)^l_{\alpha_i},(\bb_i)^l_{f_{i,i}(\alpha_i)}) \in W,$$
and similarly for every $l \in [j]$, every $i=0,\ldots,n-1$, and every $\beta_i \in A$, 
$$((\ba_{i+1})^l_{f^{-1}_{i+1,i}(\beta_i)},(\bb_i)^l_{\beta_i}) \in W.$$
Therefore, if for $i=0,\ldots,n-1$ we choose 
$$\beta_i=f_{i,i}(\alpha_i), \quad \alpha_{i+1}=f_{i+1,i}^{-1}(\beta_i)$$
we get that for every choice of $\alpha_0 \in A$, $((\ba)^l_{\alpha_0},(\bb)^l_{f(\alpha_0)}) \in \overline{W}$ as required. Hence, $\overline{W}$ is a winning strategy.

The partition $\partp$ of $\astark \cup \bstark$ is defined by the transitive closure of $\overline{W}$. That is, $\ba \in A^j$ and $\bb \in B^j$ are in the same set of $\partp$ if $(\ba,\bb) \in \overline{W}$; any two $\ba, \ba' \in A^j$ are in the same set of $\partp$ if there exists some $\bb$ such that $(\ba,\bb), (\ba',\bb) \in \overline{W}$; and similarly for any two $\bb, \bb' \in B^j$.
In particular, there is a partial isomorphism between any two tuples $\bd, \bd'$ in the same class of $\partp$. Recall that, since the universe of $\astarkbf$ is collapsed to just $\cup_{j \leq k}A^j$ (see item (\ref{obs:1pebbles}) above), we do not need to worry about defining a partition of $\con\bfa \cup \con\bfb$. It remains to show that the partition $\partp$ is equitable.

It is easily seen that any two tuples in the same class of $\partp$ are subject to the same unary constraints. In particular, if there is a partial isomorphism between $\bd$ and $\bd'$, then they are subject to the same relations in $\bfa \cup \bfb$ (accounting for unary constraints of the type $R_\emptyset$, $R \in \sig$) and they have the same repetition structure (accounting for unary constraints of the type $T_{j,S}$ for $j \leq k$, $S \subseteq[j]$).

Now we are left to deal with the binary constraints of type $\tjii$ for $j,j' \leq k$ and $\textbf{i}\in [j]^{j'}$. It is easy to see that any two tuples of arity $j$ in the same class of $\partp$ have the same number of outgoing edges of type $\tjii$ that are incident to each other class of $\partp$, since for any $j,j' \in k$, any two isomorphic tuples $\bd, \bd' \in A^j \cup B^j$, and any $\bi \in [j]^{j'}$, $\pii \bd$ and $\pii \bd'$ will still be isomorphic (equivalently, any winning strategy for the bijective $k$-pebble game is trivially closed under removing pairs of pebbles, under adding pairs of pebbles on already pebbled elements, and under permuting pairs of pebbles).

Finally, let us deal with the incoming edges of type $\tjii$. For simplicity, let us start by considering two elements of the same class of $\partp$ belonging to different structures $\bfa$, $\bfb$. That is, let $\ba \in A^{j'}$, $\bb \in B^{j'}$ for some $j' \leq k$, $(\ba,\bb) \in \overline{W}$, and let $j \leq k$ and $\bi \in [j]^{j'}$. For each partition class $\ipartp_l$ of $\partp$, we need to show that $\ba$ and $\bb$ have the same number of incoming edges labelled  $\tjii$ from $\ipartp_l$.

Notice that, if $\{\bi\}=[j]$, then there is nothing to prove (since as we pointed out above, the winning strategy is trivially closed under adding pebbles on already pebbled elements and permuting pairs of pebbles), so we may assume that $\{\bi\} \subsetneq [j]$, and for the same reasons, we may assume that $\bi$ has no repetitions and that its entries are in increasing order.
Moreover, from item (\ref{obs:4pebbles}) above, we have that $|[j]\setminus\{\bi\}|=1$. All together then we may assume that $j'=j-1$ and $\bi=\bi(j,i)$ for some $i \in [j]$. Therefore, for any $j'$-tuple $\bd$, $\pii\bd^i_d=\bd$. 

Since $(\ba,\bb) \in \overline{W}$, there exists a bijection $f:A \to B$ 
such that, for each $a \in A$, $(\ba^i_a,\bb^i_{f(a)}) \in \overline{W}$. Hence, for each $a \in A$, $\ba^i_a$ and $\bb^i_{f(a)}$ belong to the same class of $\partp$. By the definition of $\tjii$, this means that $\ba$ and $\bb$ have the same number of incoming edges labelled  $\tjii$ from each class of $\partp$. Since $|\astark|=|\bstark|$ and $\astarkbf$,  $\bstarkbf$ are disjoint components of $\astarkbf\cup\bstarkbf$, this also implies that any partition class of $\partp$ contains the same number of elements from $\astark$ and $\bstark$.

Now we still need to check that this condition holds for any two elements of the same class of $\partp$ that belong to the same structure, say $\bfa$. But this is immediate since if $\ba,\ba' \in A^{j'}$ belong to the same class of $\partp$, then there exists $\bb \in B^{j'}$ such that $(\ba,\bb)$ and $(\ba',\bb)$ belong to $\overline{W}$, and since both $\ba$ and $\bb$ and $\ba'$ and $\bb$ have the same number of incoming edges labelled $\tjii$ from each class of $\partp$, the same must hold for $\ba$ and $\ba'$. It follows that $\partp$ is equitable.
\end{proof}

\bibliography{cleanbib}

\appendix
\section{Applying the SA method exactly} \label{app:realSAmethod}

In this section we will show that the family of relaxations $\sa^k$ considered in the present paper is 
closely interleaved with the system of relaxations obtained by applying the SA method to a natural choice of initial polytope $\polytope$.

Let $\bfx$ and $\bfa$ be $\sig$-structures. We define the polytope $\polytope=\polytope\xa$ using a system of inequalities. The variables of the system are $y_{x,a}$ for each $x\in X$ and $a\in A$. Each variable must take a value in the range $[0,1]$. We remark that by fixing some arbitrary ordering on the variables $y_{x,a}$ we can represent any assignment of the variables $y_{x,a}$ with a tuple $\by\in\mathbb{R}^n$ with $n=|X|\cdot |A|$. Therefore we shall abuse notation and use $\by_{x,a}$ to refer to the value of the assignment $\by$ corresponding to variable $y_{x,a}$.

The variables are subject to the following inequalities. 
\begin{align}
& \sum_{a\in A} y_{x,a}=1 && x\in X, \label{eq:blp2} \\
& \sum_{x\in\{\bx\}} y_{x,f(x)}\leq |\{\bx\}|-1 &&
\textnormal{for each } R\in\sig, \ \bx\in R^\bfx,\, \textnormal{ and} \label{eq:blp3}\\[-8pt]
&&& f:\{\bx\}\rightarrow A \textnormal{ with } f(\bx)\not\in \RA. \nonumber
\end{align}

Note that if $h$ is a homomorphism from $\bfx$ to $\bfa$ then the assignment setting $y_{x,h(x)}=1$ for every $x\in X$ and the rest of variables to zero is feasible.

Now let $\polytope^k$, $k\geq 1$ be the sequence of polytopes obtained by applying the SA method to $\polytope$. The following result shows that the sequence of relaxations defined by $\sa^k$ and $\polytope^k$ are interleaved.

\begin{lemma} \label{le:SAkrPk}
Let $k\geq 1$ and let $r$ be the maximum arity of a relation in $\sig$. Then
\begin{enumerate}
\item If $\polytope^k\neq\emptyset$ and $r\leq k$ then $\sa^{k}$ is feasible. \label{1SAkrPk}
\item If $\sa^{k+r-1}\xa$ is feasible then $\polytope^k\neq\emptyset$. \label{2SAkrPk}
\end{enumerate}
\end{lemma}
\begin{proof}

(\ref{1SAkrPk}). Assume that $\polytope^k\neq\emptyset$ and let $\bz$ be a feasible solution of $\polytope^k_L$. We shall construct a feasible solution of $\sa^k\xa$. First, set every variable of the form $p_V(f)$ to $z_K$ where $K=\{(x,f(x)) \mid x\in V\}$. We first observe that this assignment satisfies (\ref{eq:SA1}) and (\ref{eq:SA2}). Indeed, let $U\subseteq X$ with $|U|<k$, let $f:U\to A$, and let $I=\{(u,f(u)) : u\in U\}$. Then, multiplying the equality (\ref{eq:blp2}) with $x\in X \setminus U$ by $\Pi_{i\in I} y_i$ and linearizing we obtain equality (\ref{eq:SA2}) for $U$, $f$,  and $V=U\cup\{x\}$. In this way we can obtain all equalities in (\ref{eq:SA2}) for $|U|+1=|V|$. We note here that the rest of equalities in (\ref{eq:SA2}) along all equalities in (\ref{eq:SA1}) can be obtained as a linear combination.

Secondly, let us set the rest of variables. For every $(\bx,R)\in \con{\bfx}$ and $f:\{\bx\}\to A$, set $p_{(\bx,R)}(f)$ to be $z_K$ where $K=\{(x,f(x)) \mid x\in\{\bx\}\}$ (note that we are using implicitly the fact that $r\leq k$). Then, (\ref{eq:SA3}) follows directly from (\ref{eq:SA2}). Finally, it only remains to show that (\ref{eq:SA4}) is also satisfied. Indeed, for every $f(\bx)\not\in \RA$ we obtain
equality $p_{(\bx,R)}(f) = 0$ if we multiply (\ref{eq:blp3}) by the term $\Pi_{i\in K}$ and linearize.

(\ref{2SAkrPk}). Assume that $\sa^{k+r-1}\xa$ is feasible. We construct a feasible solution $\bz$ of $\polytope^k_L$ as follows. For every $K \subseteq X \times A$ which satisfies $K=\{(x,f(x)) \mid x\in U\}$ for some $U\subseteq X$ with $|U|\leq k$ and $f:U \to A$, we set $z_K:=p_U(f)$. Otherwise, we set $z_K$ to zero. 

 Let us show that this assignment satisfies all inequalities in $\polytope^k_L$. Let 
 \begin{equation}\label{eq:cyd}
 \bc^T \bz\leq \bd
 \end{equation}
 be any inequality defining $\polytope^k_L$. Since (\ref{eq:cyd}) is obtained by multiplying an inequality which contains at most $r$ variables by a term of at most $k-1$ variables, there 
exists a set $V\subseteq X$ with $|V|\leq r+k-1$ such that for every variable $z_K$ appearing in (\ref{eq:cyd}), $V$ satisfies
$K\subseteq V\times A$. Note that, by (\ref{eq:SA1}), variables $p_V(g)$, $g:V \to A$ define a probability distribution. For every $g:V \to A$ in the support of this distribution, consider the assignment $\by^g$ that sets $\by^g_{x,a}=1$ if $x\in V$ and $b=g(x)$ and $y_{x,b}=0$ otherwise. 

Inequality (\ref{eq:cyd}) has been obtained by multiplying an inequality from (\ref{eq:blp2}) or (\ref{eq:blp3}) by a term and linearizing. We claim that in both cases, the inequality that has generated (\ref{eq:cyd}) is satisfied by $\by^g$. If the inequality generating (\ref{eq:cyd}) is $\sum_{a\in A} y_{x,a}=1$ for some $x\in X$ this follows simply from the fact that $x\in V$. Assume now that (\ref{eq:cyd}) has been generated by inequality $\sum_{x\in\{\bx\}} y_{x,f(x)}\leq |\{\bx\}|-1$. In this case
note that $\{\bx\}\subseteq V$ and then the claim follows from  (\ref{eq:SA3}) and (\ref{eq:SA4}). This finalizes the proof of the claim.

Consequently, since $\by^g$ is integral it follows that the assignment $\bz^g$ defined as $\bz^g_K=\Pi_{i\in K} \by^g_i$ satisfies (\ref{eq:cyd}). Finally, note that if we set $\alpha^g=p_V(g)$, then for every $K\subseteq V\times A$, $\bz_K$ is precisely
given by the convex combination $\sum_g \alpha^g \bz^g_K$.
\end{proof}

\end{document}